%% file: main.tex
\documentclass[sigconf]{acmart}
\usepackage[utf8]{inputenc}
\usepackage{amsmath}

\DeclareMathOperator*{\argmin}{arg\,min}
\usepackage{mathtools}
\usepackage{amssymb}
\usepackage{listings}
\usepackage{array}
\usepackage{xspace} 
\usepackage[normalem]{ulem} 
\usepackage{graphicx}
\graphicspath{ {./figures/} }
\usepackage{varwidth}
\usepackage[ruled, linesnumbered]{algorithm2e}
\usepackage{enumitem}
\usepackage{footnote}
\usepackage{tablefootnote}
\setlist{  
  listparindent=\parindent,
  parsep=0pt,
}

\newif\iftechreport
\techreporttrue 

\newif\ifsubmission
\submissiontrue 

\newif\ifrevision
\revisionfalse

\definecolor{maroon}{rgb}{0.5, 0.0, 0.0}

\ifsubmission

\else

\fi

\ifrevision

\newcommand{\newchanges}[1]{\textcolor{blue}{#1}}
\newcommand{\reviewone}[1]{\textcolor{red}{#1}}
\newcommand{\reviewtwo}[1]{\textcolor{blue}{#1}}
\newcommand{\reviewthree}[1]{\textcolor{ForestGreen}{#1}}
\newcommand{\onethree}[1]{\textcolor{Purple}{#1}}

\else

\newcommand{\newchanges}[1]{#1}
\newcommand{\reviewone}[1]{#1}
\newcommand{\reviewtwo}[1]{#1}
\newcommand{\reviewthree}[1]{#1}
\newcommand{\onethree}[1]{#1}

\fi

\long\def\comment#1{}

\renewcommand{\Re}{\mathbb{R}}

\newcommand{\remove}[1]{}
\def\mparagraph#1{\par\medskip\noindent{\textbf{#1.}}\quad\parindent 1.5em}

\newcommand{\bucket}{b}

\setcopyright{none}

\title{Combining Aggregation and Sampling (Nearly) Optimally for Approximate Query Processing}
\iftechreport
\titlenote{A version of this paper has been accepted to SIGMOD'21. This document is its associated technical report. This work is mainly done when Zechao was at the University of Chicago.}
\fi

\author{Xi Liang}
\affiliation{%
  \institution{University of Chicago}
}
\email{xiliang@uchicago.edu}

\author{Stavros Sintos}
\affiliation{%
  \institution{University of Chicago}
}
\email{sintos@uchicago.edu}

\author{Zechao Shang}
\affiliation{%
  \institution{Snowflake Computing}
}
\email{zechao.shang@snowflake.com}

\author{Sanjay Krishnan}
\affiliation{%
  \institution{University of Chicago}
}
\email{skr@cs.uchicago.edu}

\iftechreport
\else
\input{ccs.tex}
\fi

\iftechreport
    \renewcommand\footnotetextcopyrightpermission[1]{} 
\else
    \copyrightyear{2021} 
    \acmYear{2021} 
    \setcopyright{acmcopyright}\acmConference[SIGMOD '21]{Proceedings of the 2021 International Conference on Management of Data}{June 20--25, 2021}{Virtual Event, China}
    \acmBooktitle{Proceedings of the 2021 International Conference on Management of Data (SIGMOD '21), June 20--25, 2021, Virtual Event, China}
    \acmPrice{15.00}
    \acmDOI{10.1145/3448016.3457277}
    \acmISBN{978-1-4503-8343-1/21/06}
\fi

\begin{document}

\iftechreport
\pagestyle{plain} 
\else
\fancyhead{}
\fi

\input{abstract.tex}
\maketitle

\setcounter{section}{0}

\input{introduction.tex}

\input{background.tex}

\input{system}

\input{optimization}
\input{experiments.tex}

\input{conclusion.tex}
\input{ack.tex}

\bibliographystyle{abbrv}
\bibliography{refs}

\newpage
\input{appendix}
 
\end{document}

%% file: ccs.tex
\ccsdesc[500]{Information systems~Database query processing}
\ccsdesc[500]{Theory of computation~Sketching and sampling}

%% file: abstract.tex
\begin{abstract}
Sample-based approximate query processing (AQP) suffers from many pitfalls such as the inability to answer very selective queries and unreliable confidence intervals when sample sizes are small.
Recent research presented an intriguing solution of combining materialized, pre-computed aggregates with sampling for accurate and more reliable AQP.
We explore this solution in detail in this work and propose an AQP physical design called PASS, or Precomputation-Assisted Stratified Sampling.
PASS builds a tree of partial aggregates that cover different partitions of the dataset.
The leaf nodes of this tree form the strata for stratified samples.
Aggregate queries whose predicates align with the partitions (or unions of partitions) are exactly answered with a depth-first search, and any partial overlaps are approximated with the stratified samples.
We propose an algorithm for optimally partitioning the data into such a data structure with various practical approximation techniques.
\end{abstract}

%% file: introduction.tex
\section{Introduction}
There are a number of applications where exact query results are unnecessary.
For example, visualizations only require precision up to screen and human perceptual resolutions~\cite{kamat2014distributed, liu2013immens}.
Similarly, in exploratory data analysis where users may be only looking for broad trends, exact numerical results are not needed~\cite{zgraggen2016progressive}.
\reviewone{In industrial workloads, such queries are not only highly prevalent, but they are also highly resource-intensive: Agarwal et al. found that roughly 30\% of a Facebook workload consists of aggregate queries over tables larger than 1TB~\cite{agarwal2013blinkdb_}. }
Examples such as these have motivated nearly 40 years of approximate query processing (AQP) research, where a database deliberately sacrifices accuracy for faster~\cite{cormode2011synopses, kandula2016quickr} or more resource-efficient results~\cite{krishnan2015stale} for aggregate queries.

Data sampling has been the primary approach used in AQP since the research area's conception~\cite{olken1986simple}. To this day, there are new results in sampling techniques~\cite{walenz2019learning}, implementation~\cite{kandula2019experiences}, and mechanisms~\cite{zhao2018random}. The perennial research interest in data sampling stems from its generality as an approximation technique and the extensive body of literature in sampling statistics to quantify the error rate in such approximations.
However, small samples may not contain the relevant data for highly selective queries, and this can lead to confusing or misleading results.
Recent work mitigates this problem by using an anticipated query workload to prioritize sampling certain regions of the database to support such selective queries (called stratified sampling)~\cite{agarwal2013blinkdb, chaudhuri2007optimized}, or construct samples online during query execution~\cite{kandula2016quickr, walenz2019learning}.

Consequently, pure uniform sampling is rarely used on its own, and most practical sampling systems leverage workload information to ensure that the system can reliably answer selective queries~\cite{acharya2000congressional, chaudhuri2007optimized, agarwal2013blinkdb, babcock2003dynamic, chaudhuri2001overcoming, chaudhuri2001robust, ganti2000icicles, ding2016sample+}.
These systems often run offline optimization routines to materialize optimal samples to answer future queries.
If we can tolerate expensive, up-front sample materialization costs, there is a natural question of whether it is also valuable to expend resources to compute helper ``full dataset'' query results.
There have been a number of different proposals that do exactly this~\cite{jermaine2003robust, galakatos2017revisiting, peng2018aqp++, wang2014sample, krishnan2015stale}; where systems leverage pre-computed exact aggregates to help estimate the result of a future query.
Not too far off from such proposals are a number of recent approaches to use machine learning for query result estimation~\cite{hilprecht2019deepdb,yang2019deep}.

Despite the handful of research papers on the subject, we find that the theory on how to best leverage both pre-computed aggregates and sampling is limited. 
Existing work often assumes one of the sides is fixed: Galakatos et al. optimize sampling given that they have cached previously computed query exact results~\cite{galakatos2017revisiting}, or Peng et al. optimize aggregate selection given a uniform sample of data~\cite{peng2018aqp++}.
Such piecewise optimization results in an incomplete understanding of the worst-case error of the data structure.
Systems have to balance a number of complex, interlinked factors: (1) precomputation/optimization time, (2) storage space, (3) query latency, and (4) query accuracy.
Tuning such structures for a desired accuracy SLO can be significantly challenging for a user.

\begin{figure}[t]
    \centering
    \includegraphics[width=\columnwidth]{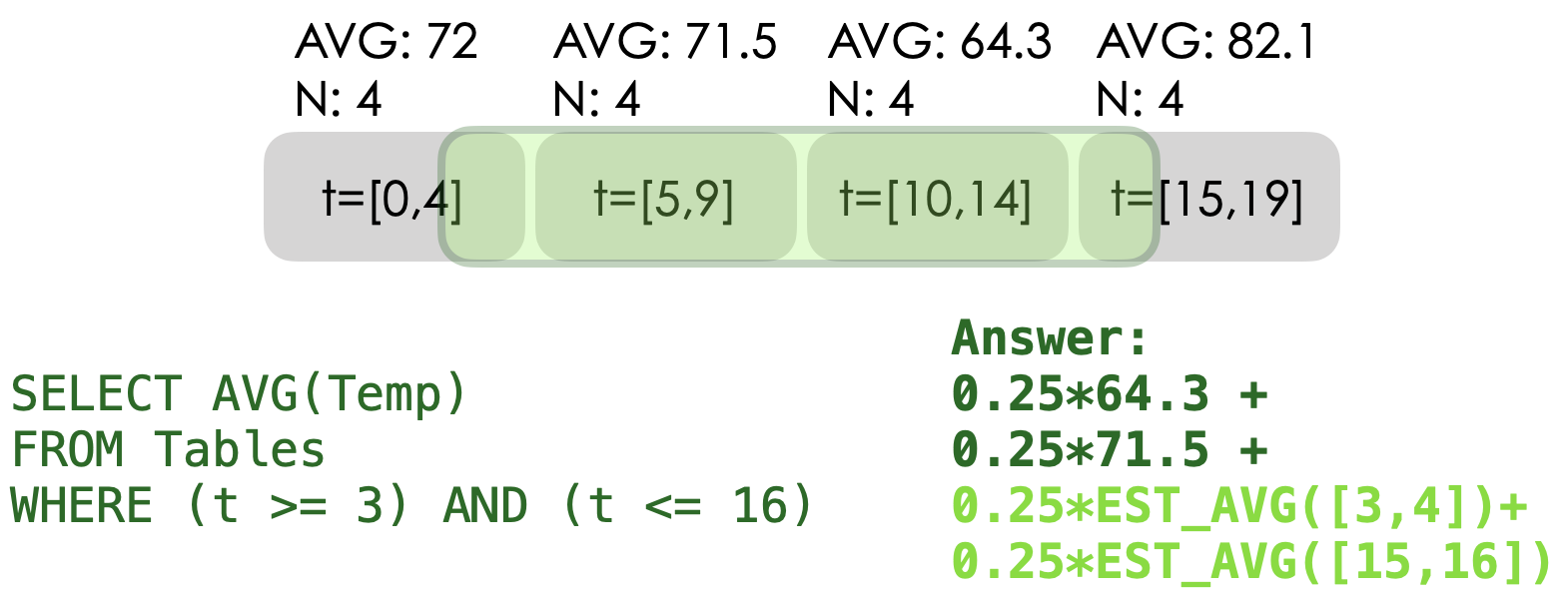}
    \caption{Temperatures collected over 20 time-steps and aggregated into 4 partitions. Partitioned aggregates can be used to decompose SUM/COUNT/AVG queries into an exact and approximate components. }
    \label{fig:mot}
\end{figure}

The joint optimization, over both sampling \emph{and} precomputation, is complex because one has to optimize over a combinatorial space of SQL aggregate queries, while accounting for the real-valued effects of sampling.
\emph{The core insight of this paper is to formalize a connection between pre-computed aggregates and stratified sampling.}
We interpret pre-computed aggregates as a sort of index that can guide sampling rather than a simple materialized cache of query results.
Figure \ref{fig:mot} illustrates a motivating example, where temperature readings over 20 time-steps are aggregated into 4 partitions.
Any new AVG query over a time range can be decomposed into two parts: an exact part where the range fully covers intersecting partitions, and an approximate part where the range partially covers a partition.
Thus, we will show a hierarchy of partitioned aggregates that can act as an efficient sample selector---determining which stratified samples of data are relevant for answering a query.
This formulation leads to a simple formula extension of stratified sampling variance, and a partitioning optimization objective that can control for the worst-case query result error from the synopsis.

We propose a new AQP data structure called PASS, or Precomputation Assisted Stratified Sampling (illustrated in Figure \ref{fig:arch}).
PASS is given a precomputation time budget and a query latency constraint, and it generates a synopsis data structure constructed of both samples and aggregates.
The more work that PASS is allowed to do upfront, the more accurate future queries are.
PASS first generates a hierarchical partitioning of the dataset (a tree).
For each partition (nodes in the tree), we calculate the SUM, COUNT, MIN, and MAX values of the partition.
Associated with the leaf nodes is a uniform sample of data from that partition (effectively a stratified sample over the leaves).
The tree-like structure acts as an index, allowing us to efficiently skip irrelevant or inconsequential partitions to the query results.
Crucially, this lends to an analytic form for query result variance for SUM, COUNT, and AVG queries with predicates.

\begin{figure}[t]
    \centering
    \includegraphics[width=\columnwidth]{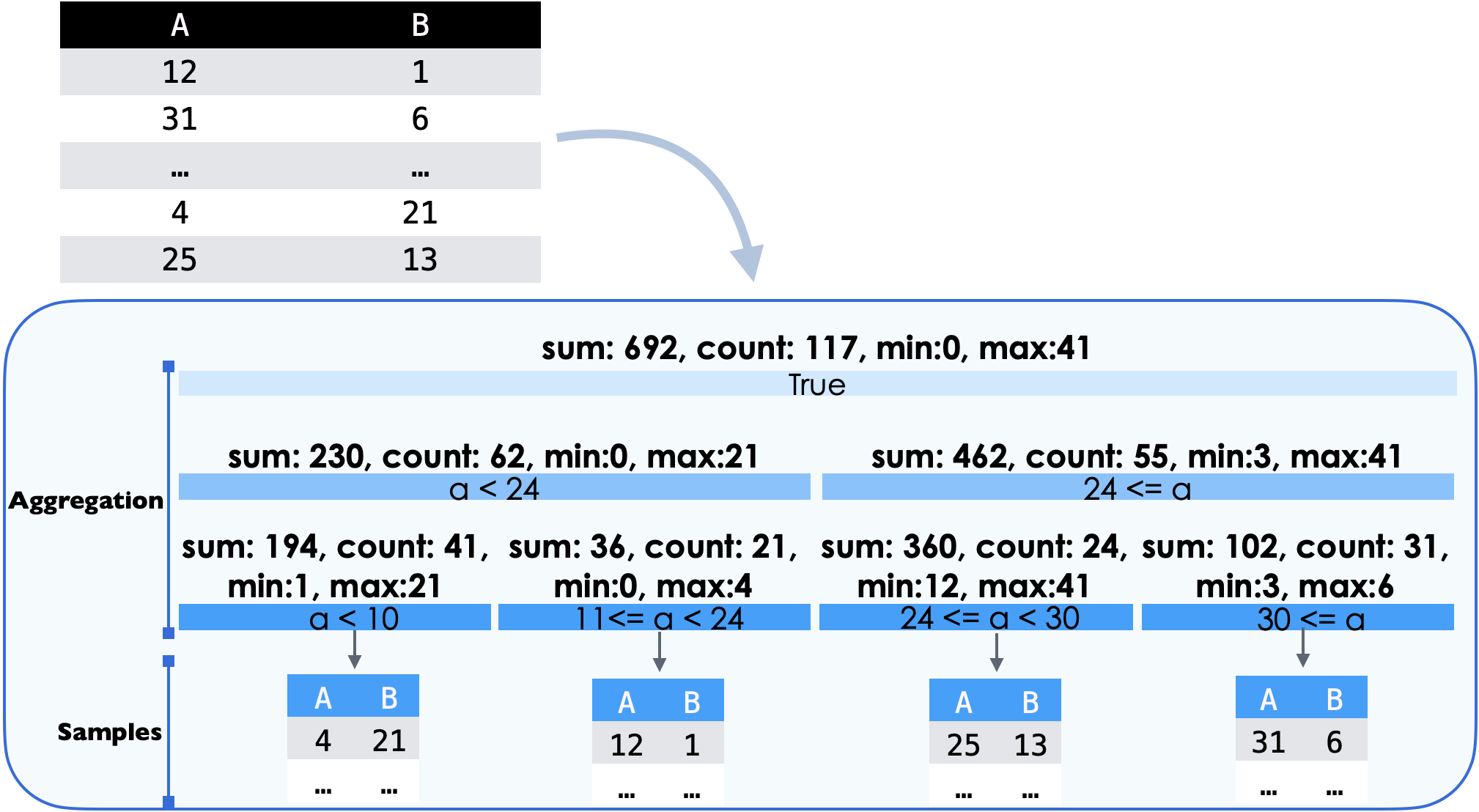}
    \caption{PASS summarizes a dataset with a tree of aggregates at different levels of resolution (granularity of partitioning). Associated with the leaf nodes are stratified samples. We present an algorithm to optimize over such a structure for fast and accurate approximate query processing.}
    \label{fig:arch}
\end{figure}

PASS gives the end user a stronger guarantee about the worst-case end-to-end accuracy than any recent ``hybrid'' AQP work~\cite{galakatos2017revisiting, peng2018aqp++, park2017database, hilprecht2019deepdb}.
In practice, we find that the data structure is often empirically beneficial compared to alternatives.
PASS is more general than the work proposed in Gan et al. and supports a wider set of queries~\cite{gan2020coopstore}.
We also find that PASS is empirically more accurate than AQP++~\cite{peng2018aqp++}, and can provably scale to much larger pre-computed sets.
A secondary benefit of PASS is that aggregate hierarchy can compute worst-case estimation error (a 100\% confidence interval) for common queries since we know the true extrema and the true cardinality of each partition (similar to~\cite{lazaridis2001progressive}).
To the best of our knowledge, no other commonly used sample-based data structure offers this benefit.

Of course, the data structure is only as good as its optimization objective. We note that PASS is sensitive to the expected workload and is more expensive to construct. 
However, we contend that PASS is a step towards a comprehensive understanding of how sampling and precomputation can be combined for fast and accurate AQP and find that if we control for query latency and the amount of precomputation, our results are generally more accurate.
In summary, we contribute:
\begin{enumerate}
    \item A new data structure for AQP called PASS which supports SUM, COUNT, AVG, MIN, and MAX queries with predicates.
    \item An optimization algorithm to generate the structure from real data that finds a partitioning that minimizes the maximum sampling variance of a set of possible expected queries.
    \item Experiments that show that PASS is often more accurate than uniform sampling, stratified sampling, and AQP++.
\end{enumerate}

%% file: background.tex
\section{Background and Problem Setup}
\label{sec:background}
We start with a simplified model: a ``one-dimensional'' approximate query processing problem.
Consider a collection of $N$ tuples representing numerical data $P = \{(c_i, a_i)\}_{i=1}^N$,
where one collects numerical value measurements $a_i$ and attributes about those measurements $c_i$ (e.g., a description about what the measurement represents); generically denoted as $A$ and $C$ respectively, when we are not interested in a particular tuple.

For example, one could have a dataset of times (attributes) and temperature measurements (values):
\begin{lstlisting}[basicstyle=\small\tt]
(00:01, 70.1C), (00:02, 70.4C),..., (15:53, 69.9C)
\end{lstlisting}
Over such a collection of data, we would like to be able to answer the following ``subpopulation-aggregate'' queries: SUM, COUNT, and AVG aggregations of the numeric measurements $A$ over subpopulations determined by filters (predicates) over the $C$.

However, in approximate query processing, we would like a sub-linear time (preferably constant time) answer with a tolerable approximation error.
This problem is fundamentally a data structure question, namely, how to summarize the collection $P$ into a synopsis that can approximately answer the desired queries.
Leading to the following more precise problem statement: \emph{derive a data structure from the collection $P$ that occupies no more than $\Tilde{\mathcal{O}}(K)$ space and can answer any subpopulation-aggregate query with approximation error provable guarantees in $\Tilde{\mathcal{O}}(K)$ time where $K$ is a user-defined parameter much less than $N$.}

\subsection{Uniform Sampling}
The simplest such data structure is a uniform sample.
From $P$, we can first sample a subset $S$ of size $K$ uniformly; that is, every tuple is sampled with equal probability. Such uniform samples of numbers have the property that the averages within the sample approximate (tend towards with bounded error) the average of the population from which the sample is derived.
So, we can approximate our desired queries by first re-formulating them as different average-value calculations.
We first define some notation:
\begin{itemize}\vspace{-.5em}
\item $t$ a tuple $(c,a)$.
\item $f(\cdot)$: a function representing any of the supported aggregates.
\item \textsf{Predicate}($t$): the predicate of the aggregate query, where \textsf{Predicate}($t$) = 1 or 0 denotes $t$ satisfies or dissatisfies the predicate, respectively.
\item $K$: the number of tuples in the sample.
\item $K_{pred}$: the number of tuples that satisfy the predicate in the sample.
\item $a$: the numerical value. 
\end{itemize}
We can reformulate SUM, COUNT, and AVG queries as calculating an average over transformed attributes:
\begin{equation}\label{eq:general-func}\small
f(S) = \frac{1}{K} \sum_{t \in S} \phi(t)
\end{equation}
where $\phi(\cdot)$ expresses all of the necessary scaling to translate the query into an average value calculation:
\begin{itemize}
\item COUNT: $\phi(t) = \textsf{Predicate}(t) \cdot N$
\item SUM: ~\, $\phi(t) = \textsf{Predicate}(t) \cdot N \cdot a$
\item AVG: ~\, $\phi(t) = \textsf{Predicate}(t) \cdot \frac{K}{K_{pred}}  \cdot a $ 
\end{itemize}

\vspace{0.5em}
In order to represent $AVG(S)$ in the form of Equation~\ref{eq:general-func}, we rewrite it to the following equivalent Equation:  
\begin{equation}
AVG(S) = \frac{1}{K}  \sum_{t \in S} \textsf{Predicate}(t) \cdot \frac{K}{K_{pred}}  \cdot a.
\end{equation}
Therefore, we have $\phi(t) = \textsf{Predicate}(t) \cdot \frac{K}{K_{pred}}  \cdot a $ for the AVG query.

\subsubsection{Error Rate}\label{subsec:resultestimation}
Let us denote the $\phi(P)$ and $\phi(S)$ as taking the transformation functions above and applying them to every tuple in $P$ or $S$ respectively.
The Central Limit Theorem (CLT) states that these empirical mean values tend towards a normal distribution centered around the population mean:
\begin{equation}\small
N(mean(\phi(P)),\frac{var(\phi(P))}{K})
\end{equation}
Since the estimate is normally distributed, we can define a confidence interval parametrized by $\lambda$ (e.g., 95\% indicates $\lambda=1.96$)\footnote{\scriptsize When estimating means of finite population there is a finite population correction factor of $FPC=\frac{N-K}{N-1}$ which scales the confidence interval.}.
\begin{equation}\small
mean(\phi(S)) \pm \lambda \sqrt{\frac{var(\phi(S))}{K}}.
\end{equation}
To understand the main pitfall of uniform sampling, notice the main scaling factor in the AVG queries is $\frac{K}{K_{pred}}$. The more selective a query is (i.e., smaller $K_{pred}$), the smaller the effective sample size is. 
If your sampling rate is $10\%$ but your predicate matches with only $1\%$ of the tuples in a database, then your effective sample size for that query is $0.1\%$! 

Not only do selective queries increase the error in your result estimates, but they also make the confidence interval estimates less reliable. Accurately estimating the variance from a very small sample is often harder than estimating the result itself since variance is a measure of spread. Furthermore, the CLT holds asymptotically and is naturally less reliable at small sample sizes.

\subsection{Stratified Sampling}\label{sec:ssamp}
Stratified sampling is one way to mitigate the effects of selective predicates.
Instead of directly sampling from $P$, we first partition $P$ into $B$ strata, which are mutually exclusive partitions defined by groupings over $C$.
Within each stratum $P_1,...,P_B$, we uniformly sample as before resulting in samples $S_1,...,S_B$.
So, instead of a single parameter $K$ which controlled the accuracy in the uniform sampling case, we have a $K_1,...,K_B$ for each stratum. The sum of all $K_i$ can be equated to the uniform sampling size to compare efficiencies $K = \sum_{i=1}^B K_i$.

The results estimation scheme in the previous section can be applied to each of the strata treating it as a full dataset.
We combine the estimates with a simple weighted average:
\[
\sum_{i=1}^B est(S_i) \cdot w_i
\]
For SUM/COUNT $w_i=1$. For AVG $w_i=\frac{N_i}{N_q}$ in strata with at least one relevant tuple to the query and $0$ otherwise; where $N_i$ is the total number of tuples in the strata, and $N_q$ is the total number of tuples in all relevant strata.
Using the algebraic properties of variance, the confidence interval can be calculated as follows:
\[
\pm \lambda \cdot \sqrt{\sum_{i=1}^B  w_i^2 \cdot V_i(q)} 
\]
where $V_i(q)$ is $\frac{var(\phi(S_i))}{K_i}$. 
Stratified sampling is really powerful when the strata correlate with predicates the user may issue. 
The variance $V_i$ within the strata might be much smaller than the variance globally.

\subsection{Stratified Aggregation}\label{sec:sagg}
Like we saw in the example in the introduction, partitioned aggregations can be used to approximate a query result.
Suppose, as in stratified sampling, we first partition $P$ into $B$ mutually exclusive partitions.
But instead of sampling from these partitions $P_1,...,P_B$, we compute the SUM, MAX, MIN, and COUNT for each partition~\footnote{For technical reasons, we assume all $a$ are positive (they can be shifted if not)}.
This data structure, a collection of partitioned aggregates, has $\mathcal{O}(B)$ values. 
We can use this data structure
to estimate query results for our desired subpopulation-aggregate queries.

For any predicate, there are three different sets of partitions:
\begin{itemize}\vspace{-.5em}
\item $R_{cover}$ : it is known that every tuple in the partition satisfies the predicate
\item $R_{partial}$ : it is possible some tuple in the partition satisfies the predicate
\item $R_{none}$ : no tuple in the partition satisfies the predicate
\end{itemize}
Since each partition $P_i$ has a $SUM(P_i)$, $MAX(P_i)$, $MIN(P_i)$, and $COUNT(P_i)$,  can use these sets to estimate the maximum possible and minimum possible value the aggregate query of interest could take.
For SUM and COUNT queries, this is easy due to their monotonic nature. We simply fully include the partial partitions in the upper bound, and omit them for the lower bound:
\[
ub = \sum_{P_i \in R_{cover}} AGG(P_i) + \sum_{P \in R_{partial}} AGG(P_i)
\]
\[
lb = \sum_{P_i \in R_{cover}} AGG(P_i)
\]
AVG queries are a little more complex to estimate since they are not monotonic. Let's define $MAX(R_{partial})$ to be the maximum of all of the max values of the partitions with partial overlap, and $MIN(R_{partial})$ to similarly be the minimum value.
A bound for the AVG query is:
\[
ub = \max\{\frac{\sum_{P_i \in R_{cover}} SUM(P_i)}{\sum_{P_i \in R_{cover}} COUNT(P)}~,~ MAX(R_{partial})\}
\]
\[
lb = \min\{\frac{\sum_{P_i \in R_{cover}} SUM(P_i)}{\sum_{P_i \in R_{cover}} COUNT(P)}~,~ MIN(R_{partial})\}
\]
The average is at most the maximum of the average of fully covered partitions and the overall max of any potentially relevant partitions (and likewise for the lower bound). This scheme is fully deterministic and it is always guaranteed that the user's result lies within those confidence intervals.

We can characterize the estimation error as $ub-lb$, and notice that in all three queries the error is a function of $R_{partial}$.
If a query predicate ``aligns'' with the partitioning (no partial overlaps), the query is answered exactly with 0 error.
This property is not guaranteed with stratified sampling, which will always have sampling error in its estimates.
However, the partial overlaps introduce ambiguity and error since we do not know how many relevant tuples match the predicate in those partitions.
In those partial overlap cases, sampling is a far more accurate estimate because the deterministic bounds are very pessimistic.


\subsection{Related Work}
Uniform sampling, stratified sampling, and stratified aggregation have dominated the AQP literature dating back to the 1980s~\cite{olken1986simple}, and we refer the readers to recent taxonomy and critique of this work~\cite{chaudhuri2017approximate}.

\noindent \textbf{Optimizing Sampling. } The pitfalls of uniform sampling are well-established, and several approaches have been proposed to optimize sampling~\cite{acharya2000congressional, chaudhuri2007optimized, agarwal2013blinkdb, babcock2003dynamic, chaudhuri2001overcoming, chaudhuri2001robust, ganti2000icicles, ding2016sample+}.
Almost all of this work relies on significant prior knowledge before query time.
Either they leverage prior knowledge of a workload~\cite{agarwal2013blinkdb, babcock2003dynamic, acharya2000congressional} or rely on auxiliary index structures~\cite{chaudhuri2001overcoming, ganti2000icicles, ding2016sample+}.
The consequence is a substantial offline optimization component that takes at least one full pass through the dataset in these AQP systems.
The proposal in this paper, PASS, is similar in that it constructs a synopsis data structure offline for accurate future query processing.
While there are AQP settings where samples are constructed online or during query processing~\cite{kandula2016quickr, hellerstein1997online, krishnan2015stale}, we believe there are a large number of data warehousing use cases where expensive upfront creation costs can be tolerated.
\reviewthree{VerdictDB~\cite{park2018verdictdb} is a recent AQP system that supports approximate query processing of general ad-hoc queries. It builds a new index \emph{scramble} by drawing samples from the original data. Given a query it uses only the sampled items in the scramble to estimate the result. They achieve fast latency with error provable guarantees. PASS also uses samples to construct a data structure in the pre-processing phase, however, it combines aggregation and stratified sampling to build a tree-based index with very low space complexity to answer aggregation queries efficiently with minimum error. VerdictDB uses more space to answer queries with high accuracy, however, it can handle more types of queries, like equi-joins. In Section~\ref{sec:exp} we compare PASS with VerdictDB over different datasets.}

\noindent \textbf{Optimizing Aggregation. } There are also similar studies of how to optimize ``binned aggregates''~\cite{jagadish1998optimal, koudas2000optimal, jagadish2001global}.  In particular, there is a highly related concept to pass of V-Optimal histograms, which are histogram buckets placed in such a way to minimize the cumulative variance~\cite{jagadish1998optimal}. In contrast, PASS is designed for cases where the goal is to aggregate one column of data based on predicates on another set of columns.  Accordingly, PASS constructs predicate partitions over the predicate columns to control the variance of the aggregation column. We further minimize the maximum variance (the worst-case error) unlike the V-Optimality condition. There are also multi-dimensional binned aggregation variants such as Lazaridis et al.~\cite{lazaridis2001progressive} (essentially a data cube for approximate query processing). While Lazardis et al. do not contribute a variance optimization, they do organize their aggregates in a hierarchical structure like PASS. 

\noindent \textbf{Hybrid AQP. } There are also a number of recent hybrid techniques that leverage precomputed ``full data'' aggregates to make sampling-based AQP more reliable. For example, Galakatos et al.~\cite{galakatos2017revisiting} cache previously computed results to augment previously constructed samples.
\reviewthree{In that way, they build an interactive scheme to handle ad-hoc queries efficiently.}
SampleClean materializes a full dataset aggregate over dirty data to mitigate sampling error~\cite{wang2014sample}.
\reviewthree{
AQP++~\cite{peng2018aqp++} precomputes a number of aggregate queries, determines query subsumption relationships to coarsely match a new query with one of those previously computed, and then, uses a uniform sample to approximate the gap. AQP++ runs a practical iterative hill-climbing heuristic to determine which aggregates to compute. We see AQP++ as the most similar proposal to PASS, but there are key differences in the two approaches. First, we propose an efficient dynamic programming algorithm with provable guarantees to find which aggregation queries (partitioning) to precompute so that the maximum error is minimized. We further organize these aggregates into a tree structure for efficient predicate evaluation.
Second, instead of using uniform sampling to approximate the gap, we apply stratified sampling only on the strata that are partially intersected by the query.}
Our experimental results find that PASS is generally more accurate for the same sample size.

\noindent \textbf{Mergeable Summaries and Partitioning. }
There is increasing discussion of data partitioning in AQP (outside of stratified sampling). Rong et al. define the PS$^3$ framework~\cite{rong2020approximate} to optimize sampling at a data partition level to avoid loading a large number of samples. We believe that the core tenets of the PS$^3$ framework are complementary to PASS and our optimization algorithm could be used as an inner routine in their framework. If our strata align with storage partitions, we could see similar benefits. We similarly see connections with Liang et al. who study constraint-based optimization for summarizing missing data~\cite{liang2020fast}. Hierarchical aggregation is also related to the work on mergeable summaries, which are synopses that can be exactly combined at different levels of granularity~\cite{gan2020coopstore, agarwal2012mergeable}.

\noindent \textbf{Learned AQP. } There are also a number of techniques that leverage machine learning for AQP. Some of the initial work in using precomputation for AQP uses Maximum Likelihood Estimates to extrapolate results to unseen queries~\cite{jermaine2003robust,jin2006new}. There also are more comprehensive solutions that train from a past query workload~\cite{park2017database} or directly build a probabilistic model of the entire database~\cite{hilprecht2019deepdb, yang2019deep}. There are also middle grounds that learn weights to direct stratified sampling
~\cite{walenz2019learning}.

%% file: system.tex
\section{Overview and Query Processing}
\label{sec:ds_and_query_processing}
These strengths and weaknesses of stratified sampling and stratified aggregation suggest an intriguing middle ground.
In the simplest version, one can create stratified samples and annotate them with precomputed partition aggregates.
To answer a supported query, we can skip all strata that are fully covered and only use the samples to estimate those partially covered strata leading to our contribution, PASS: Precomputation-Assisted Stratified Sampling.

\subsection{Usage}\label{tuning}
PASS is a synopsis data structure used for answering aggregate queries over relational data.
The user defines an \emph{aggregation column} (numerical attribute to aggregate) and a \emph{set of predicate columns} (columns over which filters will be applied).
\reviewone{The system returns an optimized data structure that can answer SUM, COUNT, AVG, MIN, and MAX aggregates over the aggregation column filtered by the predicate columns.}
\begin{lstlisting}
SELECT SUM/COUNT/AVG/MIN/MAX(A)
FROM P
WHERE Predicate(C1,...,Cd)
\end{lstlisting}
Conceptually, this is the same problem setup as described in the previous section with a dataset of $(c, a)$ tuples; where the aggregation column corresponds to $a$ and the predicate columns correspond to $c$.
A PASS data-structure is one-dimensional when there is a single predicate column and is multi-dimensional when there are a set of predicate columns. 

\reviewone{
The user specifies the following parameters: ($\tau_c$) a time limit for constructing the data structure, and ($\tau_q$) a time limit for querying the data structure.
Then, using a cost-model, our framework minimizes the maximum query error while satisfying those constraints.}
Let $\mathcal{T}$ be the set of all PASS data structures that satisfy the above constraints, and $Q$ be the set of all relevant queries to the user. We define the following optimization problem:  
\begin{equation}\label{eq:main}
    T^*=\argmin_{T\in \mathcal{T}} \max_{q\in Q}error(q,T).
\end{equation}
The details of this optimization problem are described in the next section, but for simplicity, we will only consider tree structures with a fixed fanout and ``rectangular'' partitioning conditions $x_i\leq C_i\leq y_i$ for $1\leq i\leq d$.

\subsection{Partition Trees and Samples}\label{sec:part}
A preliminary concept to understanding PASS is an understanding of multi-resolution partitioning.
A \emph{partition} of a dataset $P$ is a decomposition of $P$ into disjoint parts $P_1,...,P_B$.
Each $P_i$ has an associated partitioning condition $\psi_i$, a predicate that when applied to the full dataset as a filter retrieves the full partition.
This definition is recursive as each partition is itself another dataset. 
Partitions can be further subdivided into even more partitions, which can then be subdivided further.
This type of recursive subdivision leads to the definition of a \emph{partition tree}.

\begin{definition}[Partition Tree]
Let $\{P_i\}_1^B$ be subsets of a dataset $P$. A partition tree is a tree with $B$ nodes (one corresponding to each subset) with the following invariants: (1) every child is contained in its parent, (2) all siblings are disjoint, and (3) the union of the siblings equals the parent.
\end{definition}

Given this definition, the root of this tree is necessarily the full dataset, which we can think of as a degenerate partitioning with the condition $\psi = True$.
Siblings' conditions can be combined together with a disjunction to derive the parent, and children can be derived with a conjunction with the parent's condition.
Thus, each layer of the tree completely spans the entire dataset, but is subdivided at finer and finer granularities.

For a target query predicate $q$ and a corresponding subset of tuples that satisfy $q$ denoted by $P(q)$, we can define the \emph{coverage frontier} or a minimal set of partitioning conditions that fully covers a query.
Let $\{P_i\}_{i=1}^B$ be nodes in a partition tree. A subset $P_1,...,P_l$ of these nodes \emph{covers} a predicate $q$ if $P(q) \subset \bigcup_{i=1}^l P_i.$
A covering subset is minimal if it is the smallest subset (in terms of the number of partitions) that covers $q$.

The nature of the invariants, where disjoint children completely span their parents, described above allows us to find such a subset of nodes efficiently.
Consider the following recursive algorithm:

\begin{algorithm}
\SetAlgoLined
\SetKwFunction{mcs}{\textsc{MCF}}
\mcs{$P_i, q$}:\\
\Indp   \textbf{if}$P_i \subseteq P(q)$ or $P_i$ is a leaf: \textbf{return} $\{P_i\}$\\
        \textbf{if}$P_i \cap P(q)=\emptyset$: \textbf{return} $\{\}$ \\
$\gamma = \{\}$\\
    \textbf{for} all children $P_{i}'$ of $P_i$: $\gamma = \gamma ~\cup$ \mcs{$P_i', q$}\\
 \KwRet{$\gamma$} 
 \caption{Minimal Coverage Frontier Algorithm}
\end{algorithm}

Note that there are two types of nodes returned by the MCF algorithm above.
Either we return \emph{leaf} nodes or we return nodes that are fully contained by the query predicate.
These two types of nodes exactly correspond to the two scenarios we described in Section \ref{sec:sagg}: partial coverage and total coverage.
The leaf nodes correspond to the partial overlap case.

Such a data structure gives us a practical algorithm to scale up stratified aggregation to a large number of nodes.
Instead of a tuple-wise containment test, the base case in line 2 can be evaluated from the partitioning conditions $\psi_i$ and a description of the query predicate.
Normally, for stratified aggregation, we would have to test each of the $B$ partitions.
However, a tree facilitates faster evaluation time for selective queries. 

Suppose, we have a partition tree of $B$ nodes where every parent has a fixed number of children. Let $q$ be a query that overlaps with $\gamma$ of the leaf nodes in the partition tree.
\reviewthree{
In the worst case, for computing MCF we need to visit $O(\gamma)$ nodes in each level of the partition tree.
In this setting, if the partition tree is balanced the time-complexity for computing the MCF is $\mathcal{O}(\gamma \log B)$.}
This result can be shown by noting that the number of overlaps with leaves bounds the number of relevant nodes in any of the layers (due to the invariants) and there are $\log B$ such layers.
For selective queries, this approach is far more efficient than a linear search through all partitions.

For each of the partitions in the partition tree, we compute four aggregate statistics over all the tuples within the partition: SUM, COUNT, MIN, MAX. Note, the AVG value is implicitly calculated with SUM and COUNT.
This data structure forms the backbone of PASS: a partition tree annotated with aggregate statistics (as seen in the top half of Figure \ref{fig:arch}).
The MCF algorithm returns those partitions that are completely covered, within which we can directly leverage the pre-computed aggregate, and those that are partially covered where we will need a more sophisticated estimation scheme.

\reviewthree{It is natural to compare PASS to existing data skipping frameworks that skip irrelevant partitions of data (e.g., ~\cite{sun2014fine}). However, it is worth noting that PASS further skips partitions that are completely covered by a query predicate. This is due to the composable structure of SUM/COUNT/AVG/MIN/MAX aggregate queries supported, where they can be safely computed from partial aggregates.}

\vspace{0.5em} \noindent \textbf{Sampling: }\label{sec:leafs}
The challenge with partial coverage is that we do not exactly know the selectivity of a query within the partition.
Thus, it makes sense to leverage sampling for this estimation.
Due to the partitioning invariants, partially covered nodes will only be leaf nodes and all retrieved nodes are disjoint.
We associate with each of the leaf nodes a uniform sample of tuples \emph{within that partition}.
This per-partition sampling plan differs from other proposals such as~\cite{peng2018aqp++}, which use uniform sampling.
The entire structure is summarized in Figure \ref{fig:arch}, where the partition tree of $B$ nodes lies above and stratified samples associated with each of leaf nodes lie below.

\subsection{Query Processing}\label{sec:qp}
PASS leads to the following query processing algorithm. We present SUM, COUNT, AVG for brevity, but it is also possible to get estimations for MIN and MAX.

\vspace{0.5em} \noindent  \textbf{Index Lookup. }  Apply the MCF algorithm above to retrieve two sets of partitions: $R_{cover}$ and $R_{partial}$.

 \vspace{0.5em} \noindent \textbf{Partial Aggregation. } For each partition in $R_{cover}$, we can compute an exact ``partial aggregate'' for the tuples in those partitions. For a SUM/COUNT query $q$: $agg =  \sum_{P_i \in R_{cover}} P_i(q),$
for an AVG query, we weight the average by the relative size of the partition: $agg =  \sum_{P_i \in R_{cover}} P_i(q) \frac{N_i}{N_q}$, where $N_i$ is the size of the partition $P_i$ and $N_q$ is the total size in all relevant partitions of query $q$.

 \vspace{0.5em} \noindent \textbf{Sample Estimation. } Each partition in $R_{partial}$ is a leaf node with an associated stratified sample. We use the formula in Section \ref{sec:ssamp} to estimate their contribution. Let $S_i$ denote the sample associated with partition $P_i \in R_{partial}$: $samp = \sum_{P_i \in R_{partial}} f(S_i) \cdot w_i$

 \vspace{0.5em} \noindent \textbf{Results. } The results can be found by taking a sum of the two parts: $result = samp + agg.$

\vspace{0.5em} \noindent  \textbf{Confidence Intervals. } Since the $agg$ result is fully deterministic, the only part that needs uncertainty quanitification is $samp$. We can use the formula in Section \ref{sec:ssamp} to compute this result:
\[
\pm \lambda \cdot \sqrt{\sum_{P_i \in R_{partial}} w_i^2 \cdot V_i(q)} 
\]
where $V_i(q)$ is $\frac{var(\phi(S_i))}{K_i}$ where $K_i$ is the sample size of the sample associated with partition $P_i$. 

 \vspace{0.5em}  \noindent \textbf{Hard Bounds. } The data structure allows us to compute deterministic hard bounds on query results using formulas in Sec. \ref{sec:sagg}. 

\subsection{Other Optimizations}
To summarize, PASS associates a hierarchy of aggregates with stratified samples. This allows us to create sampling plans with a large number of strata and efficiently skip irrelevant ones.
The hierarchy further allows us to bound estimates deterministically since the extrema of each stratum are known.
The design of PASS allows for a few important optimizations that can significantly improve performance in special cases.
First, the data structure has an important special case where we can skip processing the samples even when there is a partial overlap.

\vspace{0.5em}\noindent \textbf{0 Variance Rule: } For AVG queries, we can add an additional base case to the MCF algorithm: \emph{if the node in question has 0 variance (meaning that the min value is equal to the max value), return the current node. } When answering AVG queries, 0 variance nodes (where all numerical values are the same) are equivalent to covered nodes.

Next, the data structure can also effectively compress the samples using \emph{delta encoding}. Every sampled tuple can be expressed as a \emph{delta} from its partition average. Ideally, the variance within a partition would be smaller than the variance over the whole dataset.

%% file: optimization.tex
\section{Optimizing the Partitioning}
\label{sec:OptPart}
This section describes how we optimize PASS to meet the desired error rate.

\subsection{Objective and Search Space}
The first step is to process the user-specified time constraints $\tau_c$ and $\tau_q$ into internal parameters.
As before, we consider a dataset $P$, where there is an aggregation column $A$, and a collection of predicate columns ($C_1,..,C_d$).
We calculate the maximum number of leaf nodes $k$ (which governs the construction time as we show later) allowable in the time limit $\tau_c$.
Each leaf will define rectangular partitioning condition $x_i\leq C_i\leq y_i$ for $1\leq i\leq d$.
Then, we calculate the maximum number of samples allowable in the time limit $\tau_q$.
The search space $\mathcal{T}$ is all PASS data structures with $k$ leaves
with a fixed fanout of $2^d$.

Next, we have to define a class of queries $Q$ that we care about.
While in general SQL predicates can be arbitrary, we restrict ourselves to a large class of ``sensible'' predicates.
Over this schema, we define Q to be AVG/SUM/COUNT queries in a ``rectangular region'', which returns the average value with respect to the attribute $A$ among all tuples with $x_i\leq C_i\leq y_i$ for $1\leq i\leq d$.
Our optimization framework can also support other definitions for $Q$ (hereafter called templates), but for brevity, we will focus on rectangular averages.

Given this definition of $Q$ and $\mathcal{T}$, now we describe the optimization objective.
Let $T=\{\bucket_1,\ldots, \bucket_k\}$ be the set of leaf nodes.
For a query $q\in Q$ let $T_q$ be the set of leaf nodes that partially intersect with the query predicate, we get an error formula as follows:
\[error(q,T)=\lambda\sqrt{\sum_{\bucket_i \in T_q} w_i^2 \cdot V_{i}(q) }.\]
Note from the previous section, the structure of the leaf nodes governs the estimation error of the data structure. The shape of the tree (height and fanout) only affects construction time and query latency. 
Thus, it is sufficient to optimize PASS in two steps to control for worst-case query error: first, choose an optimal partitioning of the leaf nodes, and then construct the full tree with a bottom-up aggregation.
The core of the algorithm is then to optimize the following ``flat'' partitioning:
\[R^*= \argmin_{ T \in \mathcal{T}_{leaves} } \max_{q\in Q} error(q,T),\]
where $\mathcal{T}_{leaves}$ is the family of all possible $k$ leaf nodes.

\subsection{Partitioning Algorithm Intuition}
It is worth noting that we do not have to search over all possible rectangular partitions and queries. Consider the 1D case (where the rectangles are simply intervals): all meaningful rectangular conditions will have predicate intervals defined by the attribute values of the tuples $c_i$ (any others are equivalent), and thus, \reviewtwo{there are 
${N \choose k-1}=\mathcal{O}(N^{k-1})$ possible partitions and $\mathcal{O}(N^2)$ possible query intervals.}

Thus, we make the following simplifying but practical assumptions. Note that the error in a query result is governed by the variance of items where there is partial overlap. 
So, we approximate the problem to the following search condition: control the variance of every query's single worst partial overlap.
Next, we further assume that when there is a partial overlap, this overlap is non-trivial where at least $\delta N$ tuples are relevant to the query (avoids degenerate results that could result in empty partitions).

\reviewtwo{Even in the 1D case, we want to avoid exhaustive enumeration of all the $\mathcal{O}(N^{k-1})$ partitions. In addition, we want to avoid precomputing and storing the results of all possible $\mathcal{O}(N^2)$ interval queries due to the quadratic space dependency;
if we could simply execute all $\mathcal{O}(N^2)$ queries, there is no need for a synopsis structure in the first place.}

\subsubsection{Technical Details}
The following contains technical details about the assumptions that can be skipped for brevity.
Let $\bucket_i\in R$ be a partition of partitioning $R$ (notice that the set of leaves $T$ corresponds to a partitioning $R$).
Let $N_i$ be the number of items in partition (i.e. leaf) $\bucket_i$.
Furthermore, let $R_q$ be the partitions of partitioning $R$ that intersect the query $q$ (either fully or partially). Finally, let $P_i(q)=P\cap \bucket_i \cap q$, be the set of items in bucket $i$ that are contained in query $q$, as we had in the previous section.
We assume that the valid queries of $Q$ with respect to a partitioning $R$, are the queries that intersect sufficiently many items in each partition they intersect. Precisely, for a partition $\bucket_i\in R_q$, if $N_{i,q}=|P_i(q)|$ is the number of items in $\bucket_i$ that are contained in $q$, then we assume that $N_{i,q}\geq \delta N$.

Basically, the assumption states that we only care about queries that meaningfully overlap when they do partially overlap.
Accordingly, we can define the set of ``meaningful'' queries with respect to a partitioning $R$ as $Q'=\{q\in Q\mid N_{i,q}\geq \delta N, \forall \bucket_i\in R_q\}$.
One can think of this set as rectangles whose boundaries are grounded with actual tuples in the dataset.

From Section~\ref{sec:background} we can define the ``single-partition'' $\bucket_i$ variance of AVG queries as $V_i(q)=\frac{1}{N_i}\cdot\frac{1}{N_{i,q}^2}\left[N_i\sum_{h\in P_i(q)}t_h^2 - \left(\sum_{h\in P_{i}(q)} t_h\right)^2\right]$,
while for SUM queries
$V_i(q)=\frac{1}{N_i}\cdot\left[N_i\sum_{h\in P_{i}(q)}t_h^2 - \left(\sum_{h\in P_{i}(q)} t_h\right)^2\right]$.
For COUNT queries the formula for $V_i(q)$ is identical to the formula for SUM queries with $t_h=1$ or $t_h=0$.


Based on these formulas and the above assumptions, we approximate Equation \ref{eq:main} with a simpler problem.
\reviewtwo{The high-level idea is that we focus on the problem of minimizing the maximum variance of queries that are fully contained in only one partition. We show that solving this simpler problem leads to efficient approximations for our original problem, where queries intersect multiple partitions.}

\reviewtwo{Let $Q_R^1\subseteq Q$ be the set of meaningful queries in $Q$ that intersect only one partition in partitioning $R$, and let $i_q$ be that partition.
We define a new problem, as: $R'=\argmin_{R\in \mathcal{R}} \max_{q\in Q_R^1}V_{i_q}(q)$, where $\mathcal{R}$ is the family of all possible valid partitionings.} The next lemma is shown in Appendix~\ref{apndx:proofs}.

\begin{lemma}
\label{lem:Obs}
It holds that
\[\max_{q\in Q}error(q,R')\leq \sqrt{k}\min_{R\in \mathcal{R}}\max_{q\in Q} error(q,R),\]
for SUM and COUNT queries and 
\[\max_{q\in Q}error(q,R')=\min_{R\in \mathcal{R}}\max_{q\in Q} error(q,R),\]
for AVG queries.
\end{lemma}

From Lemma~\ref{lem:Obs} it follows that any $\alpha$-approximation algorithm for the newly defined problem is also a $\alpha$-approximation for our original problem for AVG queries and a $\sqrt{k}\cdot\alpha$-approximation for SUM and COUNT queries. Notice that it is much easier to handle the newly defined problem since we can find the query with the maximum error with respect to partitioning $R$ by looking only on queries that are fully contained in partitions $b\in R$.

\reviewtwo{Finally, we note that the approximation ratio for SUM and COUNT queries in Lemma~\ref{lem:Obs} is stated in the worst case. In particular, if we guarantee that a query can partially intersect at most $\mathcal{K}$ partitions then the approximation factor is $\sqrt{\mathcal{K}}$ instead of $\sqrt{k}$. For example, in 1D a query interval can partially intersect at most $2$ partitions, so the approximation factor is $\sqrt{2}$.}

\subsection{Algorithm in 1D}
\newcommand{\MAXV}{\mathcal{M}}
\newcommand{\AMAXV}{\overline{\mathcal{M}}}
Considering the variance function for COUNT, we can show that the optimum partitioning for COUNT queries in 1D consists of equal size partitions and hence we can construct it in linear time (we provide the proof in Appendix~\ref{apndx:proofs}). Next, we mostly focus on SUM and AVG queries.

First, we consider the $1$-d case
where we have a collection of tuples $P = \{(c_i, a_i)\}_{i=1}^N$.
We have a defined query type, i.e., SUM, AVG, and we want to find the partitioning that minimizes the maximum estimation error for that query type. In this first case, we start by developing a strawman algorithm: one where we essentially enumerate all possible 1d aggregate queries over the full data.

To do so, we first sort the tuples with respect to the predicate values $c_i$. Then we define the dynamic programming table $A$ of $N$ rows and $k$ columns, where $A[i,j]$ is the optimal solution among the first $i$ data items (with respect to the predicate values) with at most $j$ partitions.
Let $\MAXV$ be function that takes as input an interval $[i_1,i_2]$ and returns the maximum variance of a query $q\in Q$ that lie completely inside $[i_1,i_2]$.

\vspace{1em} \noindent $\MAXV(i_1,i_2):$
\begin{enumerate}
    \item $\mu = 0$
    \item For all meaningful subintervals $[g,w] \subset [i_1,i_2]$:
    \begin{enumerate}
    \item Let $\phi$ be appropriately defined for the query type and the predicate $g \le C \le w$.
    \item $\mu = \max\{\mu, var(\phi(P\cap [g,w]))\}$
    \end{enumerate}
    \item return $\mu$
\end{enumerate}

\noindent Using the function $\MAXV$ we can define the recursion
$$A[i,j]=\min_{h<i}\max\{A[h,j-1], \MAXV([h+1,i])\}.$$
Notice that we can easily solve the base cases $A[i,1]$, $A[1,j]$ using the $\MAXV$ function.
In an efficient implementation of $\MAXV$ the subquery variances are computed with pre-computed prefix sums.
Since $\MAXV$ function considers $\mathcal{O}(|Q|)$ queries ($\mathcal{O}(N^2)$ assuming all possible queries). Hence, the total running time of the basic DP algorithm is $\mathcal{O}(kN^2|Q|) = \mathcal{O}(kN^4)$. 
This algorithm achieves an optimal partitioning for AVG queries and $\sqrt{2}$ approximation for SUM queries.

\vspace{0.5em} \noindent \textbf{Faster Algorithm With Monotonicity. }
The first insight ignores the query enumeration problem but focuses on the structure of the DP itself.
Suppose, we have a query $q$ completely inside a partition $\bucket_x$ and let $\bucket_y$ be another partition such that $\bucket_x\subseteq \bucket_y$.
Then, we can show that $V_x(q)\leq V_{y}(q)$. 
This statement is very intuitive: adding irrelevant data to a query can only make the estimate worse.

The exact proof of this statement depends on the type of the query. For example for the AVG query,
we have that $P_{x}(q)=P_{y}(q)$, $N_{x,q}=N_{y,q}$, and $N_x<N_{y}$. Hence,
\begin{align*}
V_x(q)&=\frac{1}{N_{x,q}^2}\left[\sum_{h\in P_{x}(q)}t_h^2 - \frac{\left(\sum_{h\in P_{x}(q)}t_h\right)^2}{N_x}\right]\\
&\leq \frac{1}{N_{y,q}^2}\left[\sum_{h\in P_{y}(q)}t_h^2 - \frac{\left(\sum_{h\in P_{y}(q)}t_h\right)^2}{N_{y}}\right]\\
&=V_{y}(q).
\end{align*}
Similar arithmetic shows that the same statement holds for SUM and COUNT queries (with a caveat in the next section).

Then, we can argue that for $h_1\leq h_2$ it holds that $A[h_1,j-1]\leq A[h_2,j-1]$ and $\MAXV([h_1+1,i])\geq \MAXV(h_2+1,i)$. 
The second inequality holds because of our last observation.
The first inequality holds because we can use the partition of $A[h_2,j-1]$ in the first $h_1$ items as a valid partition, so $A[h_1,j-1]$ can only be smaller. Because of this property, a binary search over the values of $h$ can return the value $\hat{h}$ such that $\max\{A[\hat{h},j-1], \MAXV([\hat{h}+1,i])\}$ is minimized (more details are given in Appendix~\ref{appndx:Error}). The running time of this DP algorithm is $\mathcal{O}(kN^3\log(N))$.

\subsubsection{Faster Approximate Dynamic Program}
\label{sec:dp_algo}
Next, we want to avoid having to evaluate every possible query exactly during construction.
In this step, we can leverage a uniform sample of data to estimate the query variance; we can take a uniform sample of $m$ tuples $P=\{c_i,a_i\}^m_{i=1}$ to perform the optimization.
This approximation would create partitioning rules, which we could then resample from to construct our stratified sampling.
With sampling, the complexity of the optimization algorithm would be $\mathcal{O}(k N^2 m\log(m))$.
There is an interplay between sampling and the proofs above, where we require that every query interval receives roughly the same fraction of samples. A more detailed discussion over the fraction of samples can be found in Appendix~\ref{apndx:proofs}.
Sampling avoids the circular logic in our strawman algorithm, and the optimal partitions can be found with far less computation than the exact evaluation of every possible query over the original dataset.

\newchanges{Next, instead of considering all possible query intervals in a partition, we consider a subset of such intervals to improve the running time.
For SUM (and COUNT) queries, given an interval $[i_1,i_2]$ that contains $m'$ sampled items, we consider picking only $\mathcal{O}(1)$ items
$L\subseteq P\cap [i_1,i_2]$ such that there exists an item $x\in L$ where the intervals $[i_1,x]$ and $[x,i_2]$ contain the same number of samples. Then we consider only the interval queries whose endpoints are defined by the selected items $L$. Hence, the number of query intervals we consider in a partition is still $\mathcal{O}(1)$.
This change affects line (2) in the algorithm above. Instead of considering all meaningful subintervals $[g,w] \subset [i_1,i_2]$, we consider only the intervals defined by the items in $L$ leading to a new running time of $\mathcal{O}(k m\log m)$.
We show that the maximum variance we get by checking only this subset of queries is not
smaller than the maximum variance over all meaningful queries in $[i_1,i_2]$ divided by $4$.
Our new algorithm finds a partitioning where the maximum error of a SUM (or COUNT) query is at most $2\sqrt{2}$ times the maximum error of the optimum partitioning.
For AVG queries, we show that the query with the maximum variance in a partition has length at most $2\delta m$. We precompute the variance of all possible length $\delta m$ queries (there are only $\mathcal{O}(m)$ of them) and store them in a binary search tree. Given an interval $[i_1, i_2]$ we can return the length $\delta m$ query with the maximum variance in $\mathcal{O}(\log m)$ time.
Overall the algorithm runs in $\mathcal{O}(km\log^2 m)$ time and finds a partitioning where the maximum error of an AVG query is at most $2$ times the maximum error of the optimum partition.
Due to the space limit, all algorithms and proofs are moved in Appendix~\ref{apndx:proofs}.}

To summarize, let $N$ be the total size of the dataset, $k$ be the desired number of partitions, $m$ be the number of samples drawn for optimization~\footnote{** Indicates the approximation algorithm used in the experiments}:
\begin{table}[ht!]
    \centering
    \begin{tabular}{|l|c|}
    \hline
        Naive DP & $\mathcal{O}(kN^4)$ \\
        Faster DP &  $\mathcal{O}(kN^3\log N)$ \\
        Approx Sampling &  $\mathcal{O}(kN^2 m\log m)$ \\
        Sampling + Discretization (**) &  $\mathcal{O}(k m\log m)$ \\
        \hline
    \end{tabular}
\end{table}
\vspace{-1em} \subsection{Algorithm in Higher Dimensions}
\label{sec:kd-tree}
While the DP algorithm gives an optimum partitioning in a single dimension, it is not clear how to extend this to multiple dimensions. 

For higher dimensions, we have to consider a space of partition trees that each layer defines rectangular partitioning. 
Such a space is well-parameterized by the class of balanced k-d trees.
A k-d tree is a binary tree in which leaf nodes represent $d$-dimensional points. Every non-leaf node can be thought of as a partitioning plane that divides the parent space into two parts with the same number of items. Points to the left of this hyperplane are represented by the left subtree of that node and points to the right of the hyperplane are represented by the right subtree.
We note that we can also design k-d trees with fanout $2^d$ by splitting a node over all dimensions simultaneously.

\begin{enumerate}
    \item Construct a balanced k-d tree $U$ over $\{c_i\}_{i=1}^N$.
    
    \item Start with an empty tree $U'$ initialize as the root of $U$.
    
    \item While the number of leaf nodes is less than $k$:
    \begin{enumerate}
        \item For all new leaf nodes $v$ in $U'$:
         \begin{enumerate}
            \item Apply $\MAXV$ on the items in $v$.
        \end{enumerate} 
        \item For the leaf node that contained the query with the maximum variance, add its children (from the corresponding node in $U$) to $U'$.
    \end{enumerate}
\end{enumerate}

The precomputation time is $\mathcal{O}(N\log N)$ to construct $U$. After the precomputation, the algorithm runs in $\mathcal{O}(kN^{1-1/d}|Q|)$ time.
The tree $U'$ we return has $\mathcal{O}(k)$ space and it gives an optimum partition with at most $k$ leaf nodes with respect to the k-d tree $U$.
Hence, $U'$ is optimum with respect to AVG queries and
has a $\sqrt{k}$-approximation for SUM and COUNT queries.

\reviewtwo{
A naive way to find the leaf that contains the query with the maximum variance is to consider all possible rectangular queries in the leaf, $\mathcal{O}(|Q|)=\mathcal{O}(N^{2d})$.
By extending the faster discretization method from 1D, in Appendix~\ref{apndx:proofs} we can also get a constant approximation of the query with the maximum variance in near-linear time with respect to the number of sampled items.
}

\subsection{Summary}
To summarize our analysis, we propose an optimization framework that returns partitions that control the maximum query error over a workload of hypothetical queries.
The key approximation parameters in this framework avoid having to enumerate all possible queries.
Our analysis provides a synopsis data structure and an approximation framework that has the following parameters. We summarize the effects when these parameters are increased: 
\begin{table}[ht!]
    \centering
    \begin{tabular}{c|p{2.5cm}|p{2.5cm}}
        $\textbf{Knob}$ & \textbf{Effect} & \textbf{Tradeoff}\\
        \hline
        Sample Size $K$ & + Accuracy & + Query Latency \\
        Partitions $k$ \tablefootnote{Notice that $B$ is related to the number of leaves $k$.} & + Accuracy; - Query Latency & + Init Time; + Update Cost \\
        Apx factors $m,L$ & - Worst-Case Error & + Init Time \\
    \end{tabular}
\end{table}

\emph{In our experiments, we control $K$ and $k$ across all baselines. We ensure that the query latency of the queries is roughly the same by budgeting the same amount of precomputation and samples.}

\reviewthree{
\paragraph{Extensions}
PASS can be extended to handle multiple predicates, group-by's, and categorical queries. To handle multiple predicate column sets, we construct different trees based on statistics from the workload (see notes on statistics from Facebook~\cite{agarwal2013blinkdb}). In the full version, we demonstrated the scenario of ‘workload-shifting’ in which PASS can use a synopsis that is built for one query template to solve other query templates that share one or more attributes.
Furthermore, by applying any dictionary encoding we can handle queries over categorical variables. 
Finally, PASS can handle group-bys over categorical columns, i.e. each group-by condition can be rewritten as an equality predicate condition. Then we can aggregate answers for all the selection queries to generate a final answer.
}
\reviewtwo{
\paragraph{Dynamic updates}
PASS can easily handle new insertions (or deletions) while maintaining the statistical consistency of the estimates for COUNT, SUM, and AVG queries. In particular, we can maintain samples using Reservoir sampling~\cite{vitter1985random}. Each time that a new item $t_i$ is inserted, Reservoir sampling might choose to replace a sample $t_j$ with $t_i$. Assume that $t_j$ belongs in partition $P_j$ and $t_i$ belongs in partition $P_i$. We remove $t_j$ from $P_j$ and we insert $t_i$ in $P_i$. Furthermore, we update all the statistics in the nodes from the leaf $P_i$ to the root and from the leaf $P_j$ to the root of the partition tree. In each node, we can update the statistics in $\mathcal{O}(1)$ time so the total update time depends on the height of the tree (for $d=1$ we have $\mathcal{O}(\log k)$).
However, if there are enough updates to the structure, re-optimization of the partitioning may be needed. In that case Split and Merge technique~\cite{donjerkovicdynamic, gibbons2002fast} might help to get efficient update time.
We leave this part as an interesting future problem.
}

%% file: experiments.tex
\section{Experiment Results}
\label{sec:exp}
We evaluate PASS on a number of different datasets and workloads.
We run our experiments on a Linux machine with an Intel Core i7-8700 3.20GHz CPU and 16G RAM.

\subsection{Experiment Setup}
We follow the problem setup described in Section 3.1. Given an aggregation column and a set of predicate columns, we construct a pass data structure with a specified construction time (number of leaf nodes) and query time (sampling rate).

\subsubsection{Datasets}
\textsf{Intel Wireless Dataset}: The Intel wireless dataset \cite{intelwireless} is an Internet of Things dataset with data from 54 sensors deployed in the Intel Berkeley Research lab in 2004. It contains 3 million rows, 8 columns including humidity, temperature, light, voltage as well as date and time that are measured by different sensors. In our experiments, we use the $time$ column for predicates and the $light$ column for aggregation.

\textsf{Instacart Online Grocery Shopping Dataset 2017}: The Instacart Online Grocery Shopping dataset 2017 \cite{instacart} is released by the grocery delivery service Instacart. We use the $order\_product$ table of 1.4 million entries. Each entry has 4 columns: the $order\_id$, $product\_id$, $add\_to\_cart\_order$ and $reordered$. We use the $product\_id$ column for predicate and $reordered$ column for aggregation.

\textsf{New York City Taxi Trip Records Dataset}: The New York City Taxi Trip Records dataset \cite{nyctaxi} is published by the NYC Taxi and Limousine Commission (TLC). The dataset contains the yellow and green taxi trip records including fields capturing pick-up and drop-off dates/times, pick-up and drop-off locations, trip distances, itemized fares, rate types, payment types, and driver-reported passenger counts. In our experiments, we use the \reviewtwo{7.7 million} records collected in January 2019, and unless otherwise specified, we use the $pickup\_datetime$ column for predicate and the $trip\_distance$ column for aggregation. 

\subsubsection{Metrics}
Our primary metric is \textbf{relative error} which is the difference of estimated query result and the ground truth divided by the ground truth--for fixed sample size and precomputation budget.
In all the experiments, we evaluate the median relative error over randomly selected queries.
We also measure the \textbf{confidence interval ratio} (CI Ratio) which is the ratio between the half of estimated confidence interval and the ground truth. 
This quantifies the accuracy of the confidence intervals found with each framework.
Since one advantage of PASS is that it enables aggressive and reliable data skipping, we measure the \textbf{skip rate}, which is the ratio of the tuples that are safely skipped during query processing. 

\subsubsection{Baselines}
Every baseline gets a sampling budget of $K$ and a query precomputation budget of $B$.

\begin{itemize}
    \item \textbf{Uniform Sampling (US)} Sample $K$ records from the database uniformly at random.
    \item \textbf{Stratified Sampling (ST)} Create $B$ strata, and uniformly sample $\frac{K}{B}$ records from each one. \emph{We use equal depth partitioning to construct the strata.}
    \item \textbf{AQP++}\cite{peng2018aqp++}. We implemented the hill-climbing algorithm described in the AQP++ paper. For the 1-D experiments, instead of using a BP-cube, we partition the dataset with the hill-climbing algorithm then pre-compute aggregations on the partitions to combine with the sampling results. For multi-dimensional experiments, we construct a KD-Tree which we describe in detail in Section \ref{sec:multi-d}.
\end{itemize}

Unless further specified, we use a sample rate of 0.5\%, $\lambda$=2.576 for a 99\% confidence interval and a precomputation budget of 64 queries.
The sample size is much larger than the precomputation size.
Thus, the sample size is a good proxy for query latency.



\begin{table*}[t]
    \centering
    
    \begin{tabular}{|c|c|}
        \hline
        \multicolumn{2}{|c|}{------------}\\\hline
        Approach & Mean Cost \\\hline
        US & 0.09s \\
        ST & 0.35s \\
        AQP++ & 0.8s \\
        PASS-ESS & 23s \\
        \newchanges{PASS-BSS2x} & \newchanges{23s} \\
        \newchanges{PASS-BSS10x} & \newchanges{23s} \\\hline
    \end{tabular}
    \begin{tabular}{|c|c|c|}
        \hline
        \multicolumn{3}{|c|}{COUNT}\\
        \hline
        Intel & Insta & NYC \\\hline
        0.94\%  & 1.20\% & 0.50\% \\ 
        0.16\% &  0.22\% & 0.08\\
        0.33\% & 0.37\% & 0.16\%\\  
        0.03\% & 0.038\% & 0.02\% \\ 
        \newchanges{0.12\%} & \newchanges{0.17\%} & \newchanges{0.07\%}\\ 
        \newchanges{0.06\%} & \newchanges{0.06\%} & \newchanges{0.02\%}\\\hline 
    \end{tabular}
    \begin{tabular}{|c|c|c|}
        \hline
        \multicolumn{3}{|c|}{SUM}\\
        \hline
        Intel & Insta & NYC \\\hline
        1.61\%  & 1.82\% & 1.0\% \\ 
        1.0\% &  1.27\% & 0.8\% \\ 
        0.5\% & 0.47\% & 0.2\% \\ 
        0.05\% & 0.07\% & 0.044\% \\ 
        \newchanges{0.23}\% & \newchanges{0.3\%} & \newchanges{0.16\%}\\ 
        \newchanges{0.1\%} & \newchanges{0.11\%} & \newchanges{0.07\%}\\\hline 
    \end{tabular}
    \begin{tabular}{|c|c|c|}
        \hline
        \multicolumn{3}{|c|}{AVG}\\
        \hline
        Intel & Insta & NYC \\\hline
        1.21\%  & 1.25\% & 0.87\% \\ 
        1.0\% &  1.22\% & 0.89\% \\ 
        0.4\% & 0.31\% & 0.22\% \\ 
        0.04\% & 0.057\% & 0.04\% \\ 
        \newchanges{0.2\%} & \newchanges{0.23\%} & \newchanges{0.15\%}\\ 
        \newchanges{0.08\%} & \newchanges{0.09\%} & \newchanges{0.07\%}\\\hline 
    \end{tabular}
    \caption{\newchanges{Controlling for worst-case query latency (total number of tuples processed), we demonstrate that it is possible to construct a synoposis that is highly accurate (less than $.1\%$ error) across 2000 random SUM/COUNT/AVG queries. The caveat is a high upfront optimal partitioning cost.}}
    \label{tab:full-benchmark}
\end{table*}

\begin{figure}[ht]
  \centering
  \includegraphics[width=1\linewidth]{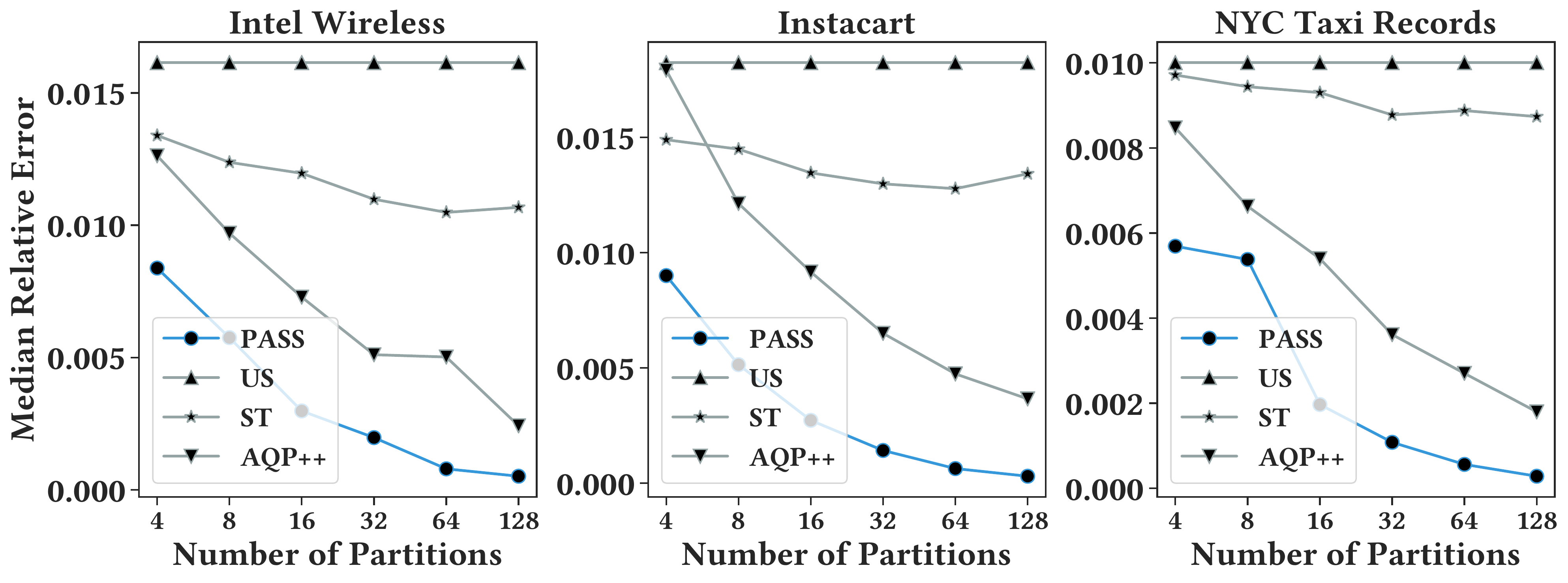}
  \caption{Median relative error of 2000 random SUM queries on the 3 real-life datasets using a varying number of partitions and a fixed sample rate of 0.5\%.}
  \label{fig:macro-re-by-k}
\end{figure}

\subsubsection{Comparing Baselines}
\newchanges{In AQP, one usually evaluates accuracy for a fixed number of sample data points processed. Since PASS couples result estimation with data skipping, controlling for the exact sample size is less intuitive. The most straightforward way to compare sampling rates across techniques is to use the effective sampling size (\textbf{ESS}) (average number of data points processed per query divided by total data points). ESS is a good metric when the main objective is to control for query latency because it basically measures the IO cost of answering an aggregate query. However, solely comparing techniques w.r.t ESS can be misleading if one is concerned about the size of the synopsis structure. Data skipping could allow one to include more samples into the synopsis if not all of them are likely to be used for any given query. Thus, we additionally include a bounded sampling size (\textbf{BSS}) comparison where techniques are restricted to a maximum number of samples. }

\subsection{Accuracy Evaluation}
\label{sec:endtoend}
Table \ref{tab:full-benchmark} illustrates the key premise of the paper: with PASS, a user pays an upfront cost for increased accuracy over future queries.
As described in the previous section, we include both a comparison in terms of ESS and BSS variants of PASS.
In the ESS case, for three datasets and randomly generated queries, a PASS synopsis (0.5\% sampling rate and 64 partitions) achieves less than a 0.1\% median relative error.
PASS is more accurate than the baselines across datasets, but does require a larger upfront optimization cost. 
\newchanges{However, ESS is an optimistic setting in certain environments with memory constraints. Thus, we additionally present the BSS versions of PASS.
In Table \ref{tab:full-benchmark}, we include PASS-BSS2x and PASS-BSS10x which are bounded to 2 times and 10 times of the online storage of uniform sampling.
Due to the data skipping, PASS-BSS2x evaluates results about 13\% faster than US/ST. 
Even with bounded storage, they still outperform other baselines significantly in terms of accuracy.
For the rest of the experiments, we will focus on the ESS setting unless explicitly mentioned as it is the most intuitive.}



\subsubsection{As a Function of Precomputation}
This construction cost is controlled by the number of partitions, which is also the amount of space allocated for aggregate precomputation.
Figure \ref{fig:macro-re-by-k} illustrates the accuracy on these three datasets for a fixed sample size of 0.5\% and a varying degree of partitions (or strata in stratified sampling). 
As the number of partitions decreases the benefits of PASS also decrease. 
\emph{PASS gives the user a new axis for control in AQP, where she can not only trade-off query latency but also data structure construction time for additional accuracy.}

\begin{figure}[ht]
  \centering
  \includegraphics[width=1\linewidth]{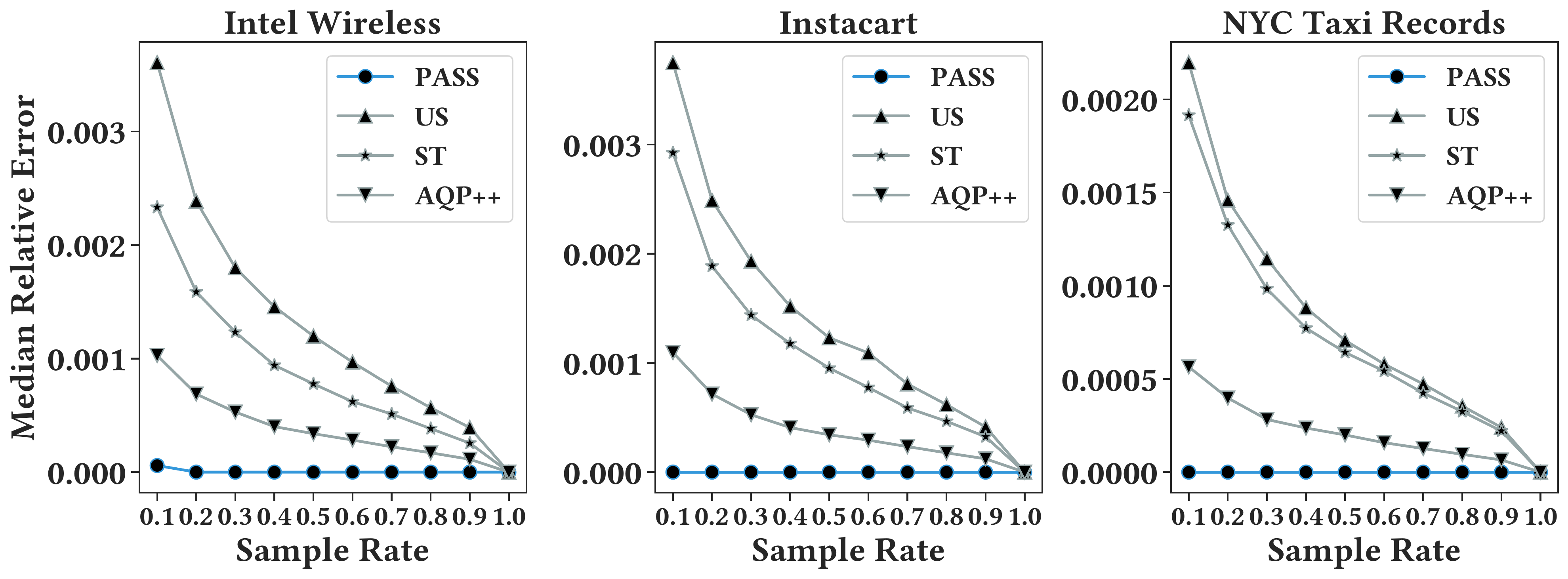}
  \caption{Median relative error of 2000 random SUM queries on the 3 real-life datasets using a varying sample rate and a fixed number of partitions of 64.}
  \label{fig:macro-re-by-sr}
\end{figure}

\begin{figure}[ht]
  \centering
  \includegraphics[width=1\linewidth]{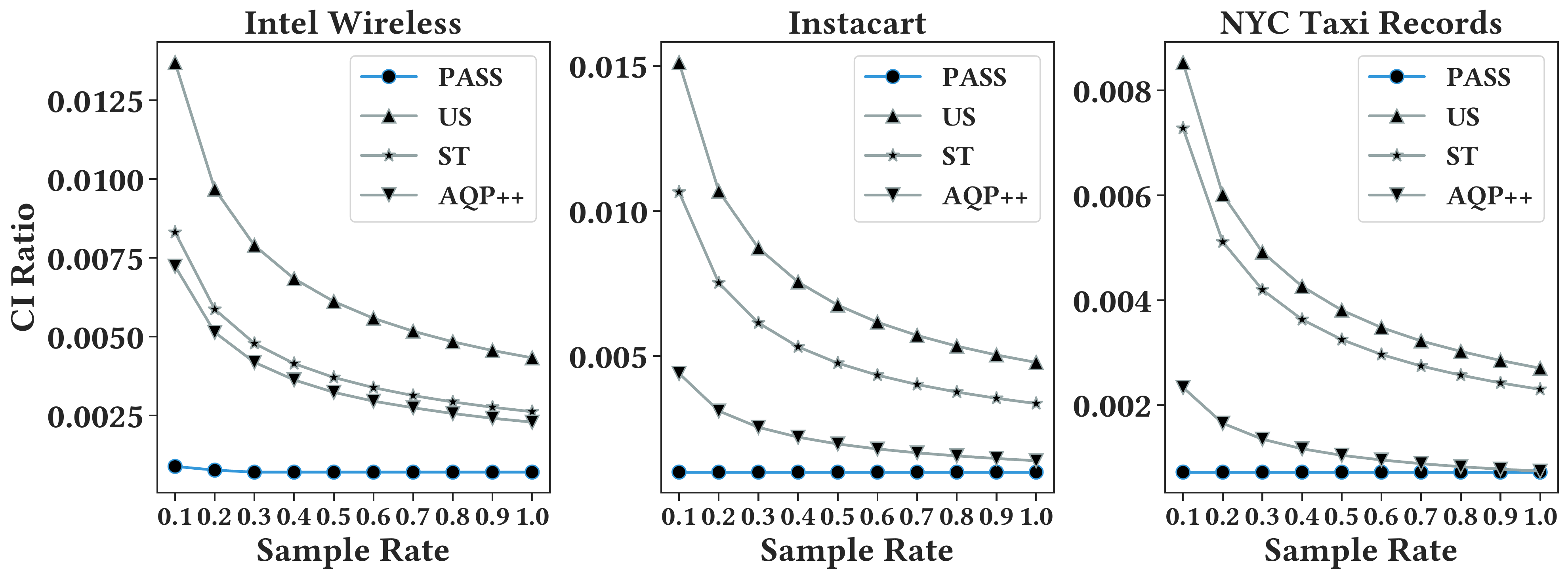}
  \caption{Median confidence interval ratio of 2000 random SUM queries on the 3 real-life datasets using a varying sample rate and a fixed number of partitions of 64.}
  \label{fig:macro-oldcir-by-sr}
\end{figure}

\subsubsection{As a Function of Sample Size}
To better understand how each baseline performs, for each dataset, we first fix the partition size to 64 and vary the sample rate from 10\% to 100\%. Figure  \ref{fig:macro-re-by-sr} shows the median relative error for 2000 random SUM queries on the Intel Wireless dataset, Instacart dataset, and the New York City Taxi Trip Records dataset. PASS outperforms other baselines starting with a 10\% sample rate. 
PASS not only returns an accurate result it also accurately quantifies this result with a confidence interval.
Figure \ref{fig:macro-oldcir-by-sr} shows the median confidence interval ratio on the three real-life datasets under the same experimental setting.
\emph{PASS is a reliable alternative to pure sampling-based synopses when expensive upfront optimization times can be tolerated.}

\subsection{Approximated Dynamic Programming Partitioning vs. Equal Partition}
In this experiment, we use different partitioning algorithms to partition the dataset, then we build a balanced binary tree bottom-up as our partition tree. The partition is later used as the strata for stratified sampling and is combined with the partition tree to solve a query as described in previously in Section \ref{sec:ds_and_query_processing}. We evaluate the approximated dynamic programming partitioning algorithm (ADP) and the equal partitioning (or equal depth, equal frequency) algorithm (EQ). We found our implementation of the hill-climbing algorithm performs very similar to the equal partitioning algorithm, so it is omitted in this experiment. 

We construct a synthetic adversarial dataset of 1 million tuples and 2 attributes. The predicate attribute contains 1 million unique values. The first 875K tuples have 0 as the value of their aggregate attribute and the last 125K tuples are generated by a normal distribution. The left plot of Figure \ref{fig:adp-cir-data1m} shows the result on 2000 random queries generated on the entire dataset and the right plot shows the result on 2000 random queries generated on the last 125K tuples. The results show that our approximated dynamic programming partitioning algorithm outperforms the EQ on the challenging queries and performs similarly on the trivial random queries. 

\begin{figure}[ht]
  \centering
  \includegraphics[width=1\linewidth]{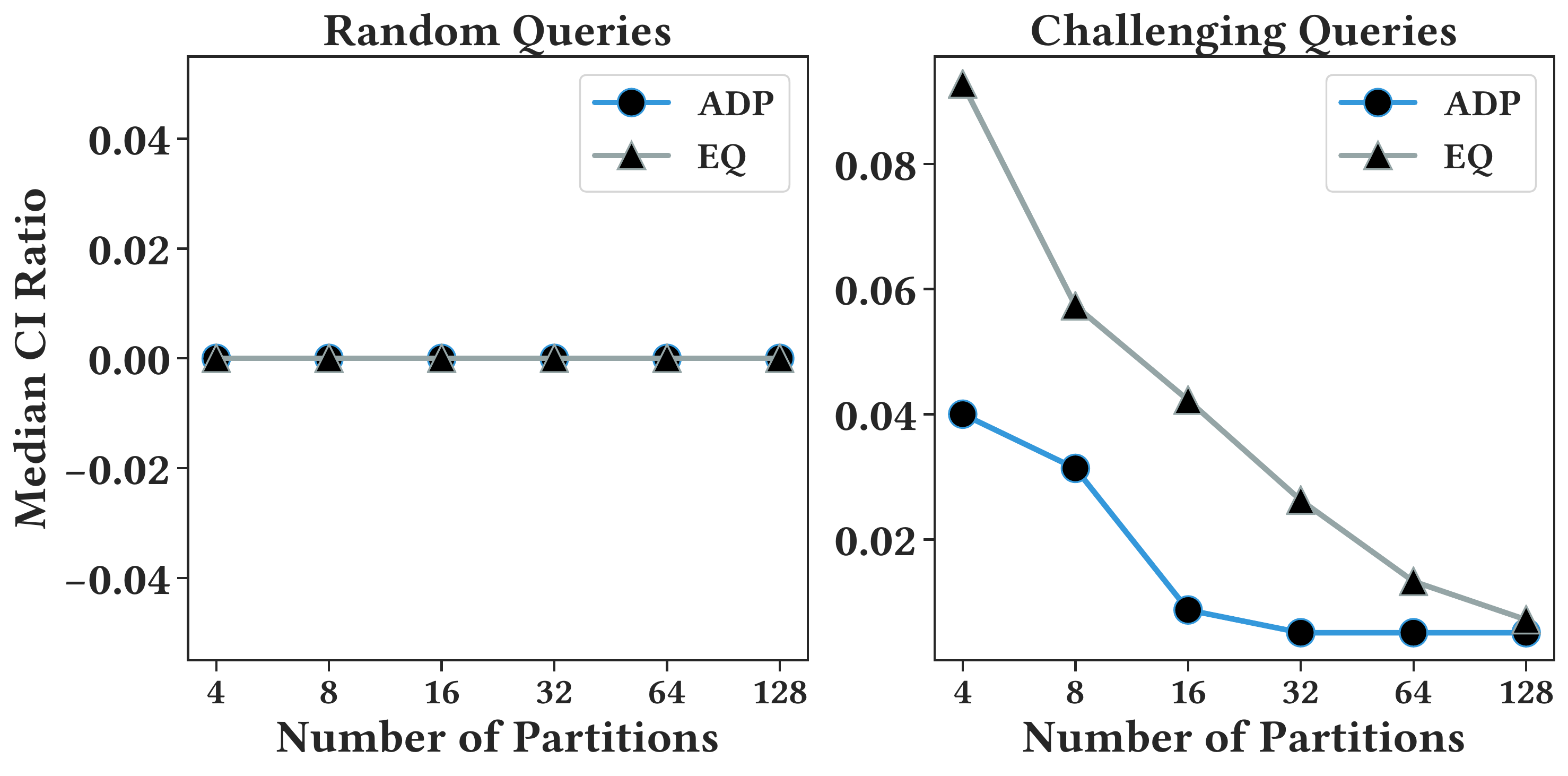}
  \caption{Median confidence interval ratio of Approximated Dynamic Programing partitioning (ADP) vs. Equal Partitioning (EQ) on a synthetic adversarial dataset.}
  \label{fig:adp-cir-data1m}
\end{figure}

Similarly, we evaluate the 3 real-life datasets by two sets of queries. For each dataset, we first randomly generate 2000 queries, then we randomly generate another 2000 challenging queries from the interval with the maximum variance identified using the fast discretization method we discussed in Section \ref{sec:dp_algo}. Figure \ref{fig:adp-cir-challenging} shows the median CI ratio on the challenging queries generated on the Intel Wireless dataset, Instacart dataset, and the New York City Taxi Records dataset. The results suggest that in most cases, ADP outperforms EQ on challenging queries.




\begin{figure}[ht]
\label{adp-eq-challenging}
  \centering
  \includegraphics[width=1\linewidth]{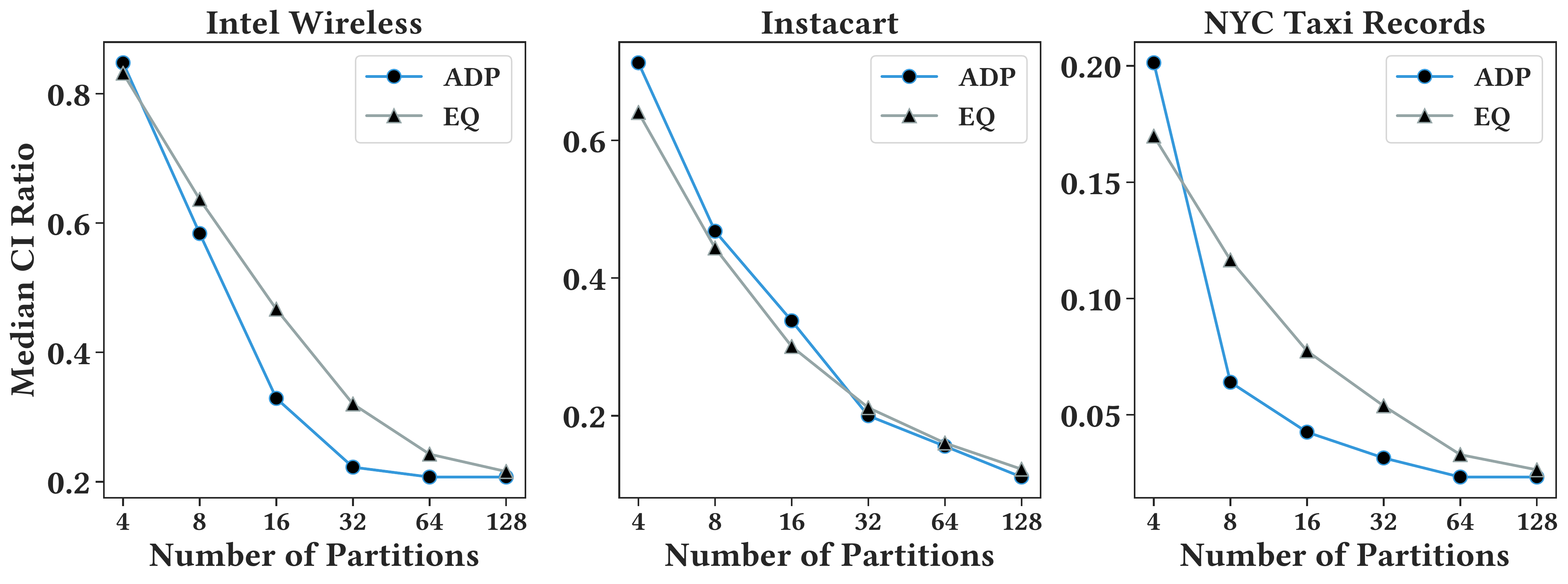}
  \caption{Median confidence interval ratio of ADP vs. EQ on challenging queries of the 3 real-life datasets.}
  \label{fig:adp-cir-challenging}
\end{figure}

\begin{table*}[ht]
    \centering
    
    \begin{tabular}{|c|}
        \hline
        \multicolumn{1}{|c|}{---}\\
        \hline
        \onethree{Approach} \\\hline
        \onethree{PASS-BSS1x}   \\ 
        \onethree{PASS-BSS2x}   \\ 
        \onethree{PASS-BSS10x}   \\ 
        \onethree{VerdictDB-10\%}  \\
        \onethree{VerdictDB-100\%} \\
        \onethree{DeepDB-10\%}  \\
        \onethree{DeepDB-100\%} \\\hline
    \end{tabular}
    \begin{tabular}{|c|c|c|}
        \hline
        \multicolumn{3}{|c|}{\onethree{Mean Cost}}\\
        \hline
        \onethree{Latency(ms)} & \onethree{Storage(MB)} & \onethree{Time(s)}\\\hline
        \onethree{24.8}  & \onethree{0.5} & \onethree{20.7} \\  
        \onethree{25.7} & \onethree{1.4} & \onethree{20.9} \\  
        \onethree{29} & \onethree{5.9} & \onethree{21.1} \\  
        \onethree{31} & \onethree{17.8} & \onethree{17} \\
        \onethree{842} & \onethree{176.8} & \onethree{49} \\
        \onethree{21} & \onethree{21.2} & \onethree{86} \\
        \onethree{22} & \onethree{61.5} & \onethree{154} \\\hline
    \end{tabular}
    \begin{tabular}{|c|c|c|c|c|c|c|}
        \hline
        \multicolumn{7}{|c|}{\onethree{Median Relative Error}}\\
        \hline
        \onethree{Intel} & \onethree{Insta} & \onethree{NYC} & \onethree{NYC-2D} & \onethree{NYC-3D} & \onethree{NYC-4D} & \onethree{NYC-5D} \\\hline
        \onethree{0.34\%} & \onethree{0.4\%} & \onethree{0.2\%} & \onethree{0.68\%} & \onethree{2.9\%} & \onethree{3.4\%} & \onethree{3.6\%} \\ 
        \onethree{0.14\%} & \onethree{0.29\%} & \onethree{0.17\%} & \onethree{0.48\%} & \onethree{2\%} & \onethree{2.1\%} & \onethree{2.26\%} \\ 
        \onethree{0.09\%} & \onethree{0.12\%} & \onethree{0.08\%} & \onethree{0.24\%} & \onethree{0.97\%} & \onethree{0.9\%} & \onethree{1.2\%} \\ 
        \onethree{90.8\%} & \onethree{90.8\%} & \onethree{90.7\%} & \onethree{90.9\%} & \onethree{90.6\%} & \onethree{90.7\%} & \onethree{90.7\%} \\ 
        \onethree{0.09\%} & \onethree{0.01\%} & \onethree{0.07\%} & \onethree{0.27\%} & \onethree{0.46\%} & \onethree{0.47\%} & \onethree{0.48\%} \\ 
        \onethree{0.9\%} & \onethree{65.8\%} & \onethree{0.9\%} & \onethree{5.2\%} & \onethree{24.6\%} & \onethree{24.8\%} & \onethree{25.6\%} \\ 
        \onethree{1.1\%} & \onethree{66.1\%} & \onethree{1.1\%} & \onethree{5.4\%} & \onethree{24.7\%} & \onethree{24.8\%} & \onethree{25.4\%} \\\hline 
    \end{tabular}
    
    \caption{\onethree{We compare the median relative error of three PASS variations with VerdictDB and DeepDB on workloads we used in previous experiments. We measure the average latency of query processing, the storage, and the construction time (training time for DeepDB) required by each approach.}}
    \label{tab:full-benchmark-verdict-deep}
\end{table*}

\subsection{Multidimensional Query Templates}
\label{sec:multi-d}
In this section, we evaluate the performance of PASS on multi-dimensional queries on the NYC Taxi dataset. Using the trip\_distance attribute as the aggregate attribute and pickup\_time, pickup\_date, PULocationID, dropoff\_date, dropoff\_time as the predicates attributes, we build 5 query templates of different dimensions where the $i_{th}$ template uses the first $i$ attribute(s) as predicate attribute(s). 

The PASS variation used in this experiment is called KD-PASS. As described in Section \ref{sec:kd-tree}, we build a KD-Tree that uses the fast discretization method to select the leaf node with the maximum variance for expansion until we reach the maximum leaf count of 1024. Also to make sure the tree is relatively balanced we limit the difference of the depth of leaf nodes to be no more than 2. At each expansion of a node, we find the median of each attribute so the fan-out factor is $2^d$. The leaf nodes of the KD-Tree forms a partition of the dataset which is used by ST for sampling and data skipping.

The baseline in this experiment is called KD-US. KD-US also uses a KD-Tree that always expands the node with the smallest depth and breaks tie randomly until we reach the maximum leaf count. The baseline then constructs a set of pre-computed aggregations based on the partition formed by the leaf nodes which is later combined with uniform sampling to generate the final answer. 

We generate 1000 queries on each query template for evaluation and results can be found in Figure \ref{fig:kdtree-cisr}. On the left plot, we show the median CI ratio of the two approaches which indicates KD-PASS outperforms KD-US. On the right figure, we plot the average skip rate of KD-PASS. We note that as we increase the dimensions of the query, the skip rate decreases. This is expected because as we increase the dimension, the partitions that are relevant to a query (thus no skipping) increase exponentially in the worst case. 

Due to the sizes of datasets, we use a maximum leaf count (i.e. number of partitions) of 1024 and a dimension of 5. Theoretically, there is no limitation in applying our framework to higher dimensions and partition sizes with proper engineering efforts.

\begin{figure}[t]
  \centering
  \includegraphics[width=1\linewidth]{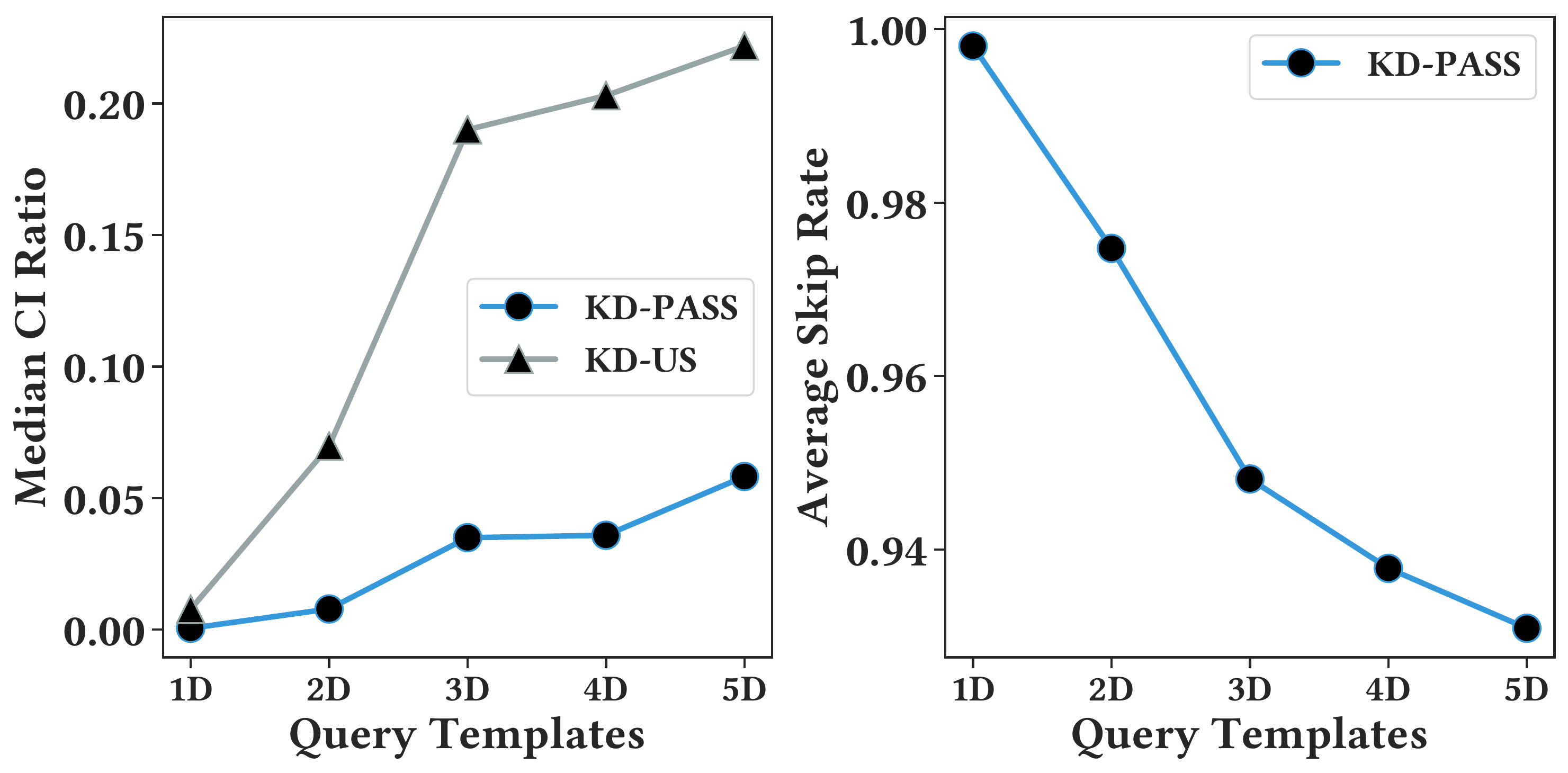}
  \caption{Multidimensional predicates on the NYC Taxi dataset. (Left) The median confidence interval ratio of KD-PASS vs. KD-US. (Right) The skip rate of KD-PASS.}
  \label{fig:kdtree-cisr}
\end{figure}

\iftechreport
\subsubsection{Workload Shift}
In this experiment, we extend the previous experiment on multi-dimensional query templates and evaluate the performance of KD-PASS and KD-US when the workload does not align perfectly with the attributes used to generate the pre-computed aggregates. We use the aggregates generated from Q2, i.e. the 2D query template, to solve all 5 templates. In this setting, the aggregates match Q2 perfectly, but it only shares 1 common attribute with Q1, 2 common attributes with Q3, Q4, and Q5.

The results shown in Figure \ref{fig:shift-cisr} is quite encouraging. In a design like AQP++ that is without data skipping, as the dimension increases, the pre-computed aggregates will be less effective in solving a query because more 'hyper-rectangles' will be intersecting with the query thus increases the area that needs to be solved by sampling, therefore, the error (variance) will increase. However, due to the unique design of PASS, as long as the query template shares one or more common attributes, even the pre-computed aggregates that are not perfectly aligned with the target query can still be used for aggressive and reliable data skipping thus increase the accuracy of the sampling and leads to an overall better result.
This can be a favorable feature in exploratory and interactive data analysis. 

\begin{figure}[ht]
  \centering
  \includegraphics[width=1\linewidth]{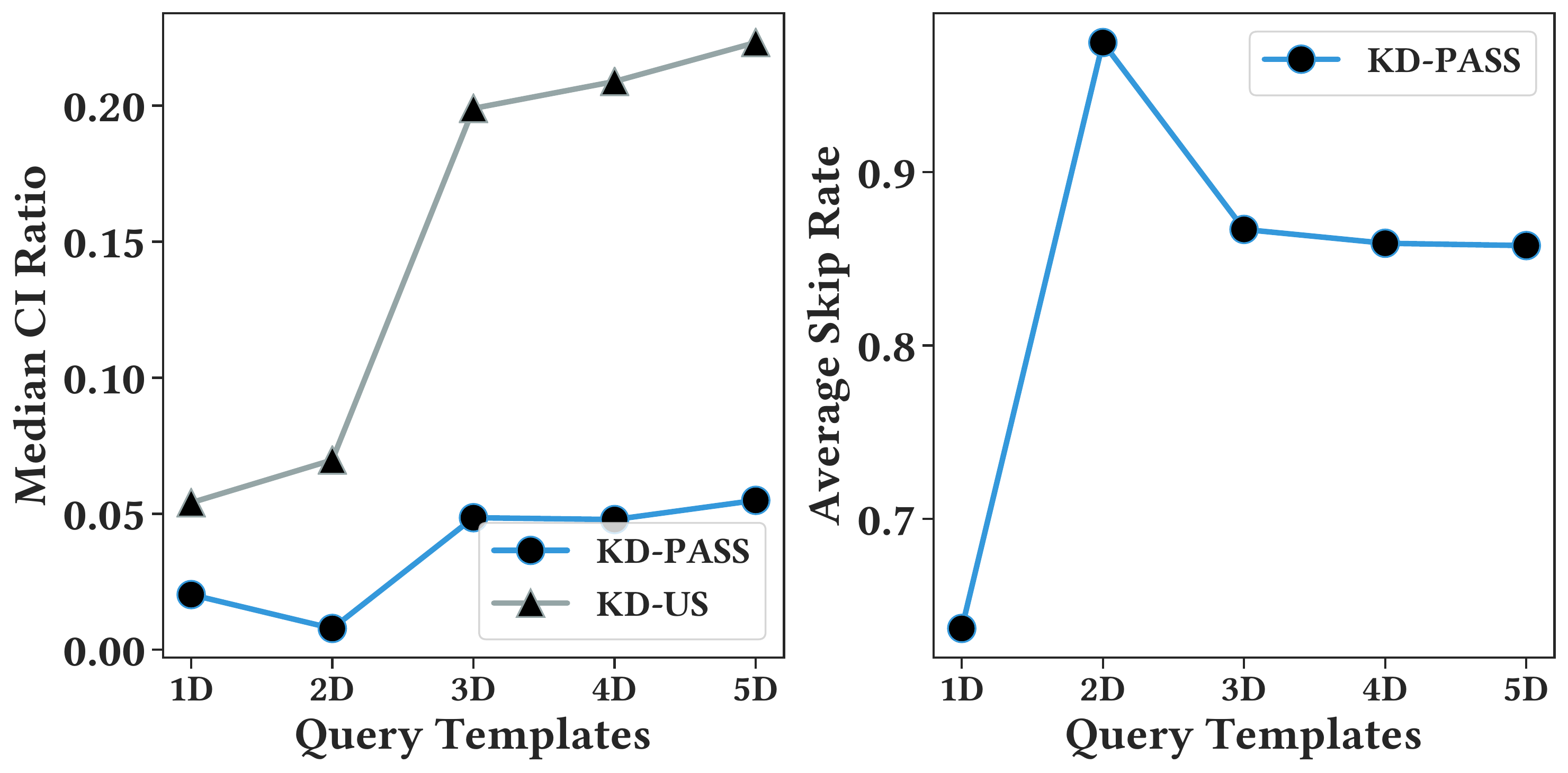}
  \caption{We use the aggregates constructed for a 2D query template to solve query templates of other dimensions. Left figure shows the Median confidence interval ratio of KD-PASS vs. KD-US on the NYC Taxi dataset. Right figure shows the percentage of tuples that are skipped by KD-PASS.}
  \label{fig:shift-cisr}
\end{figure}

\subsubsection{Preprocessing Cost}

Table \ref{tbl:cost} shows the preprocessing cost in seconds required by PASS given different numbers of partitions ($k$ in the table) on the NYC Taxi dataset. We use an optimization sample rate of 0.0025\% for the ADP algorithm to partition the dataset. As expected, the cost increases as we increase $k$ but not significantly. This is because in our implementation we cache the results of the discretization method, therefore the partitioning cost of $k$=4 does not differ a lot with $k$=64 and the difference in preprocessing cost is mostly due to the partition tree construction. \onethree{The cost of $k$=128 increases because more samples are used for a larger partition size}. The results suggest that as we increase $k$, the latency decreases and the accuracy increases, this is because a fine grain partitioning can lead to more aggressive data skipping and more efficient sampling.


\begin{table}[ht]
\begin{tabular}{|c|c|c|c|c|}
\hline
$k$ & Cost(s) & Latency(ms) & MaxLatency(ms) & MedianRE \\ \hline
4 & 16 & 14.6 & 29.2 & 0.55\% \\  
8 & 18 & 13 & 26 & 0.32\% \\  
16 & 20 & 11.6 & 23.3 & 0.18\% \\  
32 & 22 & 10.7 & 21.4 & 0.11\% \\  
64 & 25 & 8.9 & 17.8 & 0.04\% \\  
128 & 50 & 6.4 & 12.9 & 0.03\% \\ \hline 
\end{tabular}
\caption{\onethree{Preprocessing cost, mean latency, max latency and accuracy given different number of partitions.}}
\label{tbl:cost}
\end{table}
\fi

\subsection{End-to-End Comparison with Other Systems}
\label{sec:otherbaselines}
\onethree{We run extensive experiments comparing PASS to 
VerdictDB\cite{park2018verdictdb} and DeepDB\cite{hilprecht2019deepdb} on the 3 real datasets. In all of these experiments, we use the BSS mode of PASS and explicitly bound the storage size.}
\onethree{We record the mean cost in terms of query latency, storage and construction/optimization time across different workloads, and we measure the median relative error of each approach on different workloads.}
\onethree{For PASS, we build three variations using storage costs of 0.5MB, 1.4MB and 6MB; for VerdictDB, we use scrambles of a ratio of 10\% and 100\%; for DeepDB, we train models using 10\% and 100\% samples of the datasets.}

\onethree{The results in Table \ref{tab:full-benchmark-verdict-deep} show that VerdictDB-100\% generates overall the most accurate results but its storage is about the same size of the original datasets, and its latency of 842ms --- while 60\% less than MySQL --- is still much higher than PASS and DeepDB. On the other hand, the cost of VerdictDB-10\% is more favorable but the accuracy drops a lot.}
\onethree{DeepDB has the lowest latency among the three, its accuracy on the Intel Wireless dataset and the NYC 1D workload is at the same magnitude as the other two, but much worse on the Instacart dataset and the higher dimensional queries. And we also noticed that increasing the size of the training data and the storage cost of DeepDB does not necessarily improve the results.}
\onethree{The three PASS variations demonstrate a trade-off between the costs and the accuracy: as we increase the storage, the latency increase slightly and the accuracy improves. The overall performance of PASS is slightly worse than VerdictDB-100\% but we believe it is the most favorable approach given the accuracy and the costs.
}

%% file: conclusion.tex
\section{Conclusion}
While it has been proposed in previous work, we found theory around the joint use of precomputation and sampling in synopsis data structures to be limited.
The joint optimization, over both sampling \emph{and} precomputation, is complex because one has to optimize over a combinatorial space of SQL aggregate queries while accounting for the real-valued effects of sampling.
We propose an algorithmic framework that formalizes a connection between pre-computed aggregates and stratified sampling and optimizes over the joint structure.
Our results are very promising, where we see clear accuracy benefits but with the cost of initial data structure construction. 
We further show how to tradeoff cost for accuracy.

As future work, we believe AQP needs to be examined in terms of synopsis construction and maintenance costs. 
If expensive up-front costs can be tolerated then accurate results can be found.
Such a result is related to recent interest in learned models for AQP and cardinality estimation~\cite{hilprecht2019deepdb, park2018verdictdb,yang2019deep}.
While PASS is costly to construct, we believe it is a first step towards a data structure with accuracy guarantees and a variable construction time.
\reviewthree{Furthermore,
while PASS can handle multiple aggregation columns by constructing multiple trees, it
is designed considering that the number of predicate columns $d$ is a small constant (the construction time depends exponentially on $d$). For high dimensions heuristics can help improving the overall complexity, however, an interesting open problem is to efficiently build partitions with provable guarantees when $d$ is large.}

%% file: ack.tex
\section*{Acknowledgments}
This work was supported in part by the CERES Center for Unstoppable Computing, a US Army TATRC Research Award, and Intel.

%% file: appendix.tex
\appendix

\section{Missing Proofs}
\label{apndx:proofs}

\begin{proof}[Proof of Lemma~\ref{lem:Obs}]

Let $q'$ be the query with $error(q',R')=\max_{q\in Q}error(q,R')$, and let $V'=\max_{q\in Q_{R'}^1}V_{i_q}(q)$.
Since $q'$ is a valid query in $R'$, we have
$$error(q',R')=\lambda\sqrt{\sum_{b_i\in R_{q'}'}w_i^2V_i(q')}\leq \left(\sqrt{\sum_{b_i\in R_{q'}'}w_i^2}\right) \left(\lambda \sqrt{V'}\right),$$
where $R_{q'}'$ are the partitions of $R'$ that are partially intersected by $q'$.
Let $R^*$ be the optimum partitioning for Equation \ref{eq:main} and let $V^*=\max_{q\in Q_{R^*}^1}V_{i_{q}}(q)$. We note that $V'\leq V^*$, by definition.
Finally, we have $V^*\leq \max_{q\in Q}error(q,R^*)=\min_{R\in \mathcal{R}}\max_{q\in Q}error(q,R)$,
because $V^*$ corresponds to the variance of a valid query with respect to partitioning $R^*$, and $w_i=1$ for all type of queries if a query is fully contained in one partition. Hence, the maximum error in partitioning $R^*$ will always be at least $\lambda\sqrt{V^*}$.

Overall,
\[\max_{q\in Q}error(q,R')\leq \left(\sqrt{\sum_{b_i\in R_{q'}'}w_i^2}\right) \min_{R\in \mathcal{R}}\max_{q\in Q}error(q,R).\]
For SUM and COUNT queries all weights are equal to $1$. In the worst case $q'$ partially intersects $k$ partitions of partitioning $R'$ so $$\max_{q\in Q}error(q,R')\leq \sqrt{k} \min_{R\in \mathcal{R}}\max_{q\in Q}error(q,R).$$
For AVG queries it holds that $\sum_{b\in R_{q'}'}w_b^2\leq 1$, so 
$$\max_{q\in Q}error(q,R')\leq\min_{R\in \mathcal{R}}\max_{q\in Q}error(q,R),$$ and hence it follows that $$\max_{q\in Q}error(q,R')=\min_{R\in \mathcal{R}}\max_{q\in Q}error(q,R).$$

Finally, notice that if we get a partitioning $\bar{R}$ such that $$\max_{q\in Q_{\bar{R}}^1}V_{i_q}(q)\leq \alpha \cdot \min_{R\in \mathcal{R}}\max_{q\in Q_{R}^1}V_{i_q}(q),$$ for a parameter $\alpha\geq 1$, then it holds that $$\max_{q\in Q}error(q,\bar{R})\leq \sqrt{\alpha\cdot k}\cdot \min_{R\in \mathcal{R}}\max_{q\in Q}error(q,R)$$ for SUM and COUNT queries, and $$\max_{q\in Q}error(q,\bar{R})\leq\sqrt{\alpha}\cdot  \min_{R\in \mathcal{R}}\max_{q\in Q}error(q,R)$$ for AVG queries. 
\end{proof}

\begin{lemma}
\label{lem:COUNT1}
The optimum partitioning for COUNT queries in 1D consists of equal size partitions and can be constructed in near-linear time.
\end{lemma}
\begin{proof}
The error of a COUNT query $q$ in a bucket $b_i$ is defined as
$\lambda \sqrt{\frac{1}{N_i}V_i(q)}$ where $V_i(q)=[N_i\sum_{h\in P_i(q)} 1 - (\sum_{h\in P_i(q)} 1)^2]$.
We show that the query with the maximum error in a partition $b_i$ is a query that contains exactly $N_i/2$ items.
Let $X=\sum_{h\in P_i(q)} 1$. Then we have $V_i(q)=N_i\cdot X-X^2$. It is straightforward that the function $f(X)=N_i\cdot X-X^2$ is maximized for $X=N_i/2$ and the maximum error is $\frac{\lambda}{2}\sqrt{N_i}$. A query in 1D can partially intersect at most two partitions $b_i, b_j$. From the analysis above it follows that the maximum error of a query (partially) intersecting $b_i, b_j$ is $\frac{\lambda}{2}\sqrt{N_i+N_j}$ since there is always a query that contains half of items in $b_i$ and half of items in $b_j$. We claim that the optimum partitioning consists of equal size partitions $N/k$ where the maximum error is $\frac{\lambda}{2}\sqrt{2\frac{N}{k}}$. Consider any other partitioning $R$ that contains (at least) a partition with more than $N/k$ items. Let $N_i'\geq N_j'$ be the number of items in the partitions in $R$ with the maximum number of items. From the definition we have $N_i'>N/k$. Then, $N_i'+N_j'\geq N_i'+\frac{N-N_i'}{k-1}=\frac{k-2}{k-1}N_i'+\frac{N}{k-1}>\frac{k-2}{k(k-1)}N+\frac{N}{k-1}=\frac{N}{k-1}(\frac{k-2}{k}+1)= 2\frac{N}{k}$. Hence the maximum error in partitioning $R$ is $\frac{\lambda}{2}\sqrt{N_i'+N_j'}>\frac{\lambda}{2}\sqrt{2\frac{N}{k}}$.

In order to construct a partitioning with equal size partitions we sort all items and we create partitions of size $N/k$. The algorithm runs in $\mathcal{O}(N\log N)$ time.
\end{proof}

\subsection{Bounding the ratio of $N_i/n_i$} 
\label{sec:ratio}

Let $S_i(q)$ be the set of sampled items in $q$ that lie in partition $b_i$.
The error of a SUM (or COUNT) query $q$ in a partition $b_i$ is defined as $\lambda \sqrt{\frac{N_i^2}{n_i^3}\left[n_i\sum_{h\in S_i(q)}t_h^2 - (\sum_{h\in S_i(q)} t_h)^2\right]}$. In order to derive the faster approximate dynamic program we assume that the ratio $\frac{N_i^2}{n_i^2}$ is roughly the same for all partitions we check. In particular, we would like to argue that there exists a lower bound $\mathcal{L}$ and a small parameter $\beta$ such that $\mathcal{L}\leq \frac{N_i^2}{n_i^2}\leq \beta\mathcal{L}$ for all possible partitions with high probability. Notice that if this property does not hold then our algorithms still return the correct result but the running time is quadratic on $m$ (i.e., the number of samples).

We show how to bound the ratio $\frac{N_i^2}{n_i^2}$ for any possible partition the partitioning algorithm considers in any constant dimension $d$. Let $b_i$ be a rectangular partition. For $1\leq j\leq m$ let $X_j$ be random variable which is $1$ if the $j$-th sample lies in $b_i$ and $0$ otherwise. We have that $n_i=\sum_j X_j$ and let $\mu_i=E[\sum_{j}X_j]=m\frac{N_i}{N}$. From Chernoff bound we have that $Pr[(1-\gamma_i)\mu_i\leq n_i\leq (1+\gamma_i)\mu_i]\geq 1-2e^{-\gamma_i^2\mu_i/3}$, for $\gamma_i\in(0,1)$. We set $\gamma_i=\sqrt{3\frac{N}{mN_i}\ln (2m^{2d+1})}$ and we get $Pr[(1-\gamma_i)\mu_i\leq n_i\leq (1+\gamma_i)\mu_i]\geq 1-\frac{1}{m^{2d+1}}$.
There are at most $m^{2d}$ different rectangular partitions so using the union bound we get that $Pr[(1-\gamma_i)\mu_i\leq n_i\leq (1+\gamma_i)\mu_i]\geq 1-\frac{1}{m}$ for any possible valid partition $b_i$.

Next, we bound each parameter $\gamma_i$ and ratio $\frac{N_i^2}{n_i^2}$. For each partition $b_i$ we set $\gamma_i\leq 1/2$ which implies that $\sqrt{3(2d+1)\frac{N}{mN_i}\ln(2m)}\leq 1/2 \Leftrightarrow N_i\geq 12(2d+1)\frac{N}{m}\ln (2m)$. Assuming that the sample rate is $a$, i.e., $m=a\cdot N$, we get $N_i=\Omega(\frac{1}{a}\log m)$. Hence, if our partitioning algorithm considers large enough partitions (with size $\Omega(\frac{1}{a} \log m)$) we get that for all valid partitions it holds $\frac{1}{2}m\frac{N_i}{N}\leq n_i\leq \frac{3}{2}m\frac{N_i}{N}$ with probability at least $1-\frac{1}{m}$. Equivalently, if $\mathcal{L}=\frac{4}{9}\frac{N^2}{m^2}$ then $\mathcal{L}\leq \frac{N_i^2}{n_i^2}\leq 9\mathcal{L}$ for any possible large enough partition with probability at least $1-\frac{1}{m}$.

From the discussion above we can consider that the ratio $\frac{N_i^2}{n_i^2}$ is the same for all possible partitions for SUM (or COUNT when $d>1$) queries. In the worst case, with high probability, the variance will be off by a constant factor $9$. Hence, if we have an algorithm (considering the same $\frac{N_i^2}{n_i^2}$ for all possible partitions) that finds a partitioning with error at most $\alpha$ times the optimum, then the real approximation factor is $\sqrt{9}\alpha=3\alpha$, with high probability. For simplicity, we do not include the $\sqrt{9}$ factor in the approximation ratio when we state the results, however we mention it when we get the overall approximation factor.

Finally, the error for an AVG query $q$ completely inside a partition $b_i$ is defined as $\lambda\sqrt{\frac{1}{n_i\cdot |q|}\left[n_i\sum_{h\in S_i(q)}t_h^2 - (\sum_{h\in S_i(q)} t_h)^2\right]}$ (there is no dependency on $N_i$) so the sub-quadratic dynamic programming algorithm is correct in any case.

\subsection{Approximating the maximum variance in a partition}
In this section, we propose efficient algorithms for computing an approximation of the maximum variance query in an interval for $d=1$ or a rectangle for $d>1$ for SUM, COUNT and AVG queries. We consider that the values of the items are unbounded, i.e., they can be larger than $2^N$ or smaller than $1/2^N$, otherwise we can possibly get faster algorithms.

As we showed above the error of a query $q$ that lies completely inside a partition $b_i$ is defined as $$\lambda\sqrt{V_i(q)},$$
where
$$V_i(q)=\frac{N_i^2}{n_i^3}\left[n_i\sum_{h\in S_i(q)}t_h^2 - \left(\sum_{h\in S_i(q)}t_h\right)^2\right]$$
for COUNT and SUM queries and
$$V_i(q)=\frac{1}{n_i\cdot |q|^2}\left[n_i\sum_{h\in S_i(q)}t_h^2 - \left(\sum_{h\in S_i(q)}t_h\right)^2\right]$$
for AVG queries.
Since it is common in all queries we define $$\mathcal{V}_i(q)=n_i\sum_{h\in S_i(q)}t_h^2 - \left(\sum_{h\in S_i(q)}t_h\right)^2.$$
Given a partition $b_i$ (interval for $d=1$ or rectangle for $d>1$) the ratio $\frac{N_i^2}{n_i^3}$ or $\frac{1}{n_i}$ is the same for any possible query $q$ in $b_i$. Hence, 
in order to find the query with the maximum variance in a partition it is safe to define $V_i(q)=\mathcal{V}_i(q)$ for COUNT and SUM queries and $V_i(q)=\frac{1}{ |q|^2}\mathcal{V}_i(q)$ for AVG queries. It will always be clear from the context if we consider COUNT, SUM, or AVG queries. In the next sections, we also use the notation $S_i(q)$ for the set of samples in partition $b_i$ in query $q$. Finally, if $q$ is completely inside a partition $b_i$ we define $|q|=|S_i(q)|$.

Before we start we show a useful lemma that we are going to use later in all types of queries.

\begin{lemma}
\label{lem1}
Let $q$ be a query completely inside a partition $b_i$, where $|q|\leq \alpha n_i$, for a parameter $\alpha\leq 1$. Then it holds that $\mathcal{V}_i(q) \geq (1-\alpha)n_i \sum_{h\in S_i(q)}t_h^2$.
\end{lemma}
\begin{proof}
We have,
\begin{align*}
\mathcal{V}_i(q)&=n_i\sum_{h\in S_i(q)}t_h^2 - \left(\sum_{h\in S_i(q)}t_h\right)^2\\
&= (n_i-1)\sum_{h\in S_i(q)}t_h^2 - 2\sum_{w<j\in S_i(q)}t_wt_j\\
&= (n_i-|q|)\!\sum_{h\in S_i(q)}\!t_h^2 + (|q|-1)\!\sum_{h\in S_i(q)}\!t_h^2- 2\!\sum_{w<j\in S_i(q)}\!t_wt_j\\
&= (n_i-|q|)\sum_{h\in S_i(q)}t_h^2 + \sum_{w<j\in S_i(q)} (t_w-t_j)^2\\
&\geq (n_i-|q|)\sum_{h\in S_i(q)}t_h^2\\
&\geq (1-\alpha)n_i\sum_{h\in S_i(q)}t_h^2.
\end{align*}
\end{proof}

\renewcommand{\O}{\widetilde{O}}
\subsection{SUM and COUNT queries}
\label{appndx:SUM}
The key idea for finding the maximum variance query in a partition $b_i$ for SUM and COUNT queries is the same for any constant dimension $d$. Hence, we do not distinguish between $d=1$ and $d>1$. So we consider that $d$ is any constant.

Before we describe our algorithm we execute some preprocessing so we can compute efficiently $\sum_{h\in S_i(q)}t_h^2$ or $\sum_{h\in S_i(q)}t_h$ for a rectangular query $q$. In $1$-dimension we sort all samples and We compute the prefix sums $X[j]=\sum_{h\leq j}t_h^2$ and $Y[j]=\sum_{h\leq j}t_h$. In higher dimensions we construct a range tree~\cite{de1997computational} in $\mathcal{O}(m\log^{d-1} m)$ time that has space $\mathcal{O}(m \log^{d-1} m)$. Given a query rectangle $q$ the range tree can return $\sum_{h\in S_i(q)}t_h^2$ and $\sum_{h\in S_i(q)}t_h$ in $\mathcal{O}(\log^{d-1} m)$ time.

\mparagraph{Algorithm}
Let $b_i$ be a rectangle in $\Re^d$. By finding the median item (in any dimension), we split $b_i$ into two queries $q_1, q_2$ such that $q_1\cup q_2=b_i$, $S_i(q_1)\cap S_i(q_2)=\emptyset$ and $|S_i(q_1)|=|S_i(q_2)|=n_i/2$. We return $\max\{V_i(q_1), V_i(q_2)\}$.


\mparagraph{Running time}
After the preprocessing our algorithm runs in $\mathcal{O}(1)$ time for $d=1$ and $\mathcal{O}(\log^d n)$ for any constant $d>1$.

\mparagraph{Correctness}
We show the following lemma.
\begin{lemma}
\label{proof1}
Without loss of generality, let $V_i(q_1)$ be the variance that our algorithm returns and let $V_{q^*}$ be the maximum variance in partition $b_i$. Then $V_i(q_1)\geq \frac{1}{4}V_i(q^*)$.
\end{lemma}
\begin{proof}
Let $q_j$ (for $j=1$ or $j=2$) be the query such that
$j=\arg\max_{v=1, 2} \sum_{h\in S_i(q_v)}t_h^2$.
We have that $V_i(q_1)\geq V_i(q_j)$.
Furthermore, $|S_i(q_j)|=n_i/2$ so \footnote{For simplicity we assume that $n_i$ is even. The results also hold when $n_i$ is odd.} from Lemma~\ref{lem1} we have $V_i(q_j)=n_i \sum_{h\in S_i(q_j)}t_h^2 - (\sum_{h\in S_i(q_j)}t_h)^2\geq \frac{n_i}{2}\sum_{h\in S_i(q_j)}t_h^2$.
In addition, notice that $\sum_{h\in S_i(q_j)}t_h^2\geq \frac{1}{2}\sum_{h\in S_i(q^*)}t_h^2$. This holds because: If $q^*\subseteq q_1$ or $q^*\subseteq q_2$ then $\sum_{h\in S_i(q_j)}t_h^2\geq \sum_{h\in S_i(q^*)}t_h^2$. If $q^*$ intersects both of $q_1$ and $q_2$ then by definition $\sum_{h\in S_i(q_j)}t_h^2\geq \frac{1}{2}\sum_{h\in S_i(q^*)}t_h^2$.
Overall we have,
\begin{align*}
V_i(q_1)&\geq V_i(q_j)\\
&=n_i\sum_{h\in S_i(q_j)}t_h^2 - \left(\sum_{h\in S_i(q_j)}t_h\right)^2\\
&\geq \frac{n_i}{2}\sum_{h\in S_i(q_j)}t_h^2\\
&\geq \frac{n_i}{4}\sum_{h\in S_i(q^*)}t_h^2\\
&\geq \frac{1}{4}\left[n_i\sum_{h\in S_i(q^*)}t_h^2 - \left(\sum_{h\in S_i(q^*)}t_h\right)^2\right]\\
&=\frac{1}{4}V_i(q^*).
\end{align*}
\end{proof}

\subsection{AVG queries}
\label{appndx:AVG}
We first show an important property for AVG queries.
\begin{lemma}
\label{lemAVG}
The AVG query with the largest variance $q$ in a partition $b_i$ contains less than $2\delta m$ sampled items, i.e., $|S_i(q)|< 2\delta m$.
\end{lemma}
\begin{proof}
We show it by contradiction. Assume that this is not the case so $S_i(q)\geq 2\delta m$. Let $q_1, q_2$ be queries with $|S_i(q_1)|=|S_i(q_2)|=|S_i(q)|/2\leq n_i/2$, where $q=q_1\cup q_2$ and $q_1\cap q_2=\emptyset$. Without loss of generality $\sum_{h\in S_i(q_1)}t_h^2\geq \sum_{h\in S_i(q_2)}t_h^2$ so we have that $\sum_{h\in S_i(q_1)}t_h^2\geq \frac{1}{2}\sum_{h\in S_i(q)}t_h^2$.

Since $q$ is the query with the largest variance we have $$V_i(q)>V_i(q_1)\Leftrightarrow\frac{1}{|q|^2}\mathcal{V}_i(q)> \frac{1}{|q_1|^2}\mathcal{V}_i(q_1)\Leftrightarrow \mathcal{V}_i(q)> 4\mathcal{V}_i(q_1).$$ Since $|S_i(q_1)|\leq n_i/2$, from Lemma~\ref{lem1} we have that $4\mathcal{V}_i(q_1)=4[n_i\sum_{h\in S_i(q_1)}t_h^2 - (\sum_{h\in S_i(q_1)}t_h)^2]\geq 2n_i\sum_{h\in S_i(q_1)}t_h^2$. Recall that $\sum_{h\in S_i(q_1)}t_h^2\geq \frac{1}{2}\sum_{h\in S_i(q)}t_h^2$ so we conclude that $n_i\sum_{h\in S_i(q)}t_h^2 - (\sum_{h\in S_i(q)}t_h)^2>\sum_{h\in S_i(q)}t_h^2$, which is a contradiction.
\end{proof}

Using the result of Lemma~\ref{lemAVG} we can only focus on queries that contain at most $2\delta m$ items. We design two algorithms; the first one works in $1$-dimension and the second one works in any constant dimension $d$.

All our algorithms for AVG are correct assuming that the number of sampled items in a partition $b_i$ is at least $2\delta m$, i.e., $n_i\geq 2\delta m$. Otherwise we can assume that the maximum variance is $0$ because of the small number of samples in a partition.


\mparagraph{Algorithm for $d=1$}
We construct a binary search tree $T$ over the sampled items with respect to their constraint attribute. Given an interval $[l,r]$ we can use the tree $T$ to find a partition of the sampled items in $[l,r]$ with $\mathcal{O}(\log m)$ subsets, in $\mathcal{O}(\log m)$ time. Let $u$ be a node of $T$ and $T_u$ be the set of sampled items in the area (interval) of $u$ and let $n_u=|T_u|$. For a node $u$ we store the sample $t_g\in T_u$ with the largest value $\sum_{g-\delta n+1\leq h\leq g}t_h^2$. In other words, $t_g$ is the sample such that the query that contains exactly $\delta m$ samples with right endpoint $t_g$ has the largest sum of squares of the samples' values among all samples in $T_u$.

All samples with the largest values in each node of $T$ can be computed in $\mathcal{O}(m)$ time in a bottom-up way (assuming that the samples are sorted after the pre-computation).
Overall, the index we construct has $\mathcal{O}(m)$ space and can be constructed in $\mathcal{O}(m\log m)$ time.

Assume that we are given a partition $b_i=[l,r]$ in 1D and we want to find the query with the maximum variance in it. We run a search on $T$ using the interval $[l+\delta m-1,r]$. Let $\mathcal{U}$ be the set of $\mathcal{O}(\log m)$ canonical nodes of $T$ that cover all sampled items in $[l+\delta m-1,r]$.
We get the sample $t_g$ with the largest value $\sum_{g-\delta m+1\leq h\leq g}t_h^2$ among all nodes in $\mathcal{U}$.
We return the query $q'=[t_{g-\delta m+1}, t_{g}]$ with variance $V_i(q')$.

\paragraph{Running time}
After the preprocessing phase in  $\mathcal{O}(m\log m)$ time, given a partition $b_i$ we return $q'$ in $\mathcal{O}(\log m)$ time.

\paragraph{Correctness}
We show the next lemma.
\begin{lemma}
\label{lem:AVG}
Let $q^*$ be the query with the largest variance in partition $b_i$ and let $q'$ be the query returned by our algorithm. It holds that $V_i(q')\geq \frac{1}{4}V_i(q^*)$.
\end{lemma}
\begin{proof}
Our index finds the query $q'$ with the largest value $\sum_{h\in S_i(q')}t_h^2$ among all queries in $b_i$ with length $\delta m$.
Since $n_i\geq 2\delta m$ we have that $|S_i(q')|\leq n_i/2$ so from Lemma~\ref{lem1}
$\mathcal{V}_i(q')\geq \frac{n_i}{2}\sum_{h\in S_i(q')}t_h^2$.
In addition, notice that $\sum_{h\in S_i(q')}t_h^2\geq \frac{1}{2}\sum_{h\in S_i(q^*)}t_h^2$, since $|S_i(q^*)|<2\delta m$ and there exist two queries (intervals) with length $\delta m$ completely inside $b_i$ that cover $q^*$. Hence, similarly to Lemma~\ref{proof1} we get:
\begin{align*}
V_i(q')&=\frac{1}{|q'|^2}{\mathcal{V}_i(q')}\\
&=\frac{1}{|q'|^2}\left[n_i\sum_{h\in S_i(q')}t_h^2 - \left(\sum_{h\in S_i(q')}t_h\right)^2\right]\\
&\geq \frac{1}{|q'|^2}\frac{n_i}{2}\sum_{h\in S_i(q')}t_h^2\\
&\geq \frac{1}{|q'|^2}\frac{1}{4}n_i\sum_{h\in S_i(q^*)}t_h^2\\
&\geq \frac{|q^*|^2}{4|q'|^2}\frac{1}{|q^*|^2}\left[n_i\sum_{h\in S_i(q^*)}t_h^2 - \left(\sum_{h\in S_i(q^*)}t_h\right)^2\right]\\
&\geq \frac{1}{4}V_i(q^*).
\end{align*}
\end{proof}

\mparagraph{Algorithm for $d>1$}
We show two algorithms on finding the maximum variance AVG query on a rectangular partition $b_i$.

\paragraph{First algorithm}
For simplicity we describe the algorithm in $2$-dimensions, however it can be generalized to any dimension.
For preprocessing we construct a range tree over the sampled points.
Let $b_i$ be a rectangle in $2$ dimensions.
We project all sampled points in $b_i$ on the $y$ axis and we sort them with respect to $y$.
We also sort the points with respect to the $x$ coordinate. Let $t.x$ and $t.y$ be the $x$ and $y$ coordinate of sample $t$, respectively.
For each sample $t$ in $b_i$ we do the following: Let $w=2^i$ for integer $i=0,\ldots, \log m$. Let $t'$ be the $w$-th point below $t$ with respect to the $y$ coordinate. We run a binary search with respect to the $x$ axis to the right side of $t.x$ to find the item $\hat{t}$ such that the rectangle defined by the opposite corners $(t.x, t.y)\times (\hat{t}.x, t'.y)$ contains exactly $\delta m$ points. Similarly, we find the point $\bar{t}$ at the left side of $t.x$ such that the rectangle defined by the opposite corners $(\bar{t}.x, t.y)\times (t.x, t'.y)$ contains exactly $\delta m$ points. If there is not such a rectangle that contain $\delta m$ points (this can happen if there are not enough items in $b_i$ to get such a rectangle) then we proceed as follows: Without loss of generality assume that there is not such rectangle at the right side of $t.x$. Then let $\tilde{t}$ be the item with the largest $x$ coordinate. By running binary search we find the item $\hat{t}$ at the left side of $t.x$ such that the rectangle defined by the opposite corners $(\hat{t}.x, t.y)\times(\tilde{t}.x,t'.y)$ contains $\delta m$ items. If such rectangle does not exist then let $\hat{t}$ be the item with the smallest $x$ coordinate. By running binary search over the $y$ coordinates we find the item $p$ such that the rectangle defined by the opposite corners $(\hat{t}.x,t.y)\times(\tilde{t}.x, p.y)$ has $\delta m$ items. If such rectangle does not exist we skip $t$ and we continue with the next item.
For each item $t$ we also repeat the symmetric process above $t.y$, .i.e, considering the $w=2^i$-th item above $t.y$ for integer $i=0,\ldots, \log m$.
Let $\mathcal{C}$ be the set of all rectangles with size $\delta m$ that our algorithm considers.
Among all rectangles we process, we store and update the rectangle $q$ with the largest value $\sum_{h\in q}t_h^2$.
In the end, among $\mathcal{C}$, we return the rectangle (query) $q'$ with the largest value $\sum_{h\in q'}t_h^2$.

It takes $\mathcal{O}(m\log m)$ to construct the range tree. Then given a partition $b_i$, we visit each point and for each $w$ we run a binary search with respect to the $x$ coordinate. For each step of the binary search we run a query on the range tree. The range query takes $\mathcal{O}(\log m)$ and we execute $\mathcal{O}(\log m)$ steps of the binary search. There are $\mathcal{O}(\log m)$ possible values of $w$ so in total our algorithm runs in $\mathcal{O}(m\log^3 m)$ time. It also follows that $|\mathcal{C}|=O(m\log ^2 m)$.

The correctness follows from the fact that for every possible rectangle $r$ in $2$ dimensions that contains at least $\delta m$ samples and most $2\delta m$ samples (for example $q^*$) there exist at most four
(possibly overlapping) rectangles in $\mathcal{C}$ that completely cover $r$. Hence from the proof of Lemma~\ref{lem:AVG} we get that 
$V_i(q')\geq \frac{1}{8}V_i(q^*)$.

\paragraph{Second algorithm}
Next, we show a simpler algorithm that works in any constant dimension $d$.
For a partition $b_i$ we construct a modified k-d tree so that each leaf contains $\delta m$ items.
We can do it by start constructing a balanced k-d tree. If we find a node that contains exactly $\delta m$ items we stop the recursion in this path. If a node contains less than $2\delta m$ and more than $\delta m$ items we create two leaf nodes where each of them has exactly $\delta m$ items (a sample might lie in two leaves). Let $\mathcal{U}$ be the leaf nodes of the tree. Notice that $|\mathcal{U}|=O(\frac{1}{\delta})$. For each leaf node $u$ we compute $s_u=\sum_{h\in u}t_h^2$, i.e., the sum of squares of values of items in $u$. We return the query $q'$ that corresponds to the leaf $\arg\max_{u\in \mathcal{U}}s_u$.



Given a partition $b_i$ we can return $q'$ in $\mathcal{O}(m\log m)$ time (similar to constructing a k-d tree).
From the k-d tree construction we know that any rectangle can intersect at most $\mathcal{O}(\frac{1}{\delta^{1-1/d}})$ canonical nodes so the optimum query $q^*$ intersects at most $\mathcal{O}(\frac{1}{\delta^{1-1/d}})$ nodes. 
From the proof of Lemma~\ref{lem:AVG} we get that $V_i(q')\geq \frac{\delta^{1-1/d}}{2}V_i(q^*)$.

If we spend $\mathcal{O}(m\log^{d-1} m)$ pre-processing time to construct a range tree we can decide how to split a node of the k-d tree in $\mathcal{O}(\log^{d} m)$ time. In that case, after the pre-processing phase, given $b_i$ we can find $q'$ in $\mathcal{O}(\frac{1}{\delta^{1-1/d}}\log^d m)$ time.

\subsection{Approximating the overall error}
\label{appndx:Error}
In the previous section, we proposed efficient algorithms for computing an approximation of the maximum variance query in a partition $b_i$.
In this section we show what is the overall approximation error we get from the dynamic programming algorithm (for $d=1$) or the k-d tree construction (for $d>1$).
Throughout this section we assume that given a partition $b_i$ with maximum variance $V_i^*$, we can get $V_i\geq \alpha V_i^*$, for a parameter $\alpha<1$, in $\mathcal{O}(\mathcal{H})$ time (after near-linear pre-processing time).
For simplicity we assume that the ratio $\frac{N_i^2}{n_i^2}$ for any possible partition is the same. In the end we describe how the approximation is affected when the ratio is bounded as in Section~\ref{sec:ratio}.

\mparagraph{Approximation for $d=1$}

From the proof of Lemma~\ref{lem:COUNT1} we can argue that the optimum partitioning for COUNT queries in 1D consists of partitions of the same number of samples, i.e, $n_i=\frac{m}{k}$ for each of the $k$ partitions. Hence the dynamic programming algorithm and the approximation of the maximum variance procedure are not needed for COUNT queries in 1D. The optimum partitioning can be found in $\mathcal{O}(m\log m)$ time.


Next, we focus on SUM and AVG queries. We consider the faster dynamic programming algorithm that uses a binary search to find the next partition.
In particular for $A[i,j]$ we run a binary search in the interval $[1,i]$. For each value $x\in [1,i]$ we compare $A[x-1,j-1]$, $V([x,i])$, where $V([x,i])$ is the $\alpha$ approximation of the maximum variance in the interval (partition) $[x,i]$. If $A[x-1,j-1]=V([x,i])$ then we construct the partition $[x,i]$ and we set $A[i,j]=V([x,i])$. If $A[x-1,j-1]>V([x,i])$ we continue the binary search in $[1,x-1]$. If $A[x-1,j-1]<V([x,i])$ we continue the binary search in $[x+1,i]$. In the end, if $A[x-1,j-1]<V([x,i])$ and $A[x,j-1]>V([x+1,i])$ then we set $A[i,j]=\min\{V([x,i]), A[x,j-1]\}$.

\begin{lemma}
For $d=1$, the dynamic programming algorithm returns a partitioning $R$ such that $V(R)\leq \frac{1}{\alpha}V(R^*)$,
where $R^*$ is the optimum partitioning. The algorithm runs in $\mathcal{O}(km\mathcal{H}\log m)$ time.
\end{lemma}
\begin{proof}
Let $V^*[i,j]$ be the variance of the optimum partitioning of the first $i$ items with $j$ partitions.
Let $V[i,j]$ be the variance of the partitioning found by our dynamic programming algorithm among the first $i$ items with $j$ partitions.
Let $V([l,r])$ be the maximum variance returned by our approximation algorithm in the interval $[l,r]$ and let $V^*([l,r])$ be the real maximum variance. It holds that $V([l,r])\geq \alpha V^*([l,r])$.

We show that $V[i,j]\leq \frac{1}{\alpha}V^*[i,j]$. In particular we show that $\alpha V[i,j]\leq A[i,j]\leq V^*[i,j]$.
We show the result by induction on the number of partitions and the number of objects. For the base case consider $A[i,1]$ for any $i\leq m$. By definition we have that $V[i,1]=V^*[i,1]$ and $A[i,1]\geq \alpha V[i,1]$.

Assume that for all $h\leq i$ it holds that $\alpha V[h,j-1]\leq A[h,j-1]\leq V^*[h,j-1]$. We show that $\alpha V[i,j]\leq A[i,j]\leq V^*[i,j]$.

The dynamic programming algorithm runs a binary search to find the left endpoint of the $j$-th partition. Let $[h,i]$ be the $j$-th partition of the optimum partitioning among the first $i$ items with $j$ partitions. Ideally we want the binary search to find $h$, however due to the approximation of the maximum variance this is not always possible. 
Without loss of generality assume that $V^*[h-1,j-1]\geq V^*([h,i])$ so $V^*[i,j]=V^*[h-1,j-1]$. Equivalently, the proof holds if we assume that $V^*[i,j]=V^*([h,i])$.
Notice that $V^*[g,j-1]\leq V^*[h-1,j-1]$ for $g\leq h-1$ and $V^*([g,i])\geq V^*([h,i])$ for $g\leq h$. Similarly, $V^*[g,j-1]\geq V^*[h-1,j-1]$ for $g>h-1$ and $V^*([g,i])\leq V^*([h,i])$ for $g>h$. 
Let $h^-$ be the smallest index such that $\alpha V^*([h^-,i])\leq V^*[h^--1,j-1]$. Symmetrically we define $h^+$ as the largest index such that $\alpha V^*[h^+-1,j-1]\leq V^*([h^+,i])$.
We note that the binary search will always return an index $g$ such that $h^--1\leq g\leq h^+$. This holds because $V([g,i])\geq \alpha V^*([g,i])>V^*[g-1,j-1]\geq A[g-1,j-1]$ for all $g<h^-$. Equivalently we argue for $g>h^+$.

From the discussion above we have that the binary search will find an index $g\in[h^--1, h^+]$ such that $V([g,i])\geq A[g-1,j-1]$ and $A[g,j-1]\geq V([g+1,i])$. We study two cases.

If $A[g,j-1]\leq V([g,i])$ then the binary search decides that $g+1$ is the left endpoint of the $j$-th partition among the first $i$ items and sets $A[i,j]=A[g,j-1]$. From the hypothesis we have $A[g,j-1]\geq \alpha V[g,j-1]$ and $A[g,j-1]\geq V([g+1,i])\geq \alpha V^*([g+1,i])$. Hence, $A[i,j]=A[g,j-1]\geq \alpha \max\{V[g,j-1], V^*([g+1,i])\}=\alpha V[i,j]$.
If $g\in [h,h^+)$ then $V([g,i])\leq V^*[i,j]$ so $A[g,j-1]\leq V([g,i])\leq V^*[i,j]$. It is easy to see that it is not possible that $g=h^+$. If $g\in [h^--1, h-1]$ then $A[g,j-1]\leq V^*[i,j]$, by definition. So in any case $\alpha V[i,j]\leq A[i,j]\leq V^*[i,j]$. 

If $V([g,i])\leq A[g,j-1]$ then the binary search decides that $g$ is the left endpoint of the $j$-th partition among the first $i$ items and sets $A[i,j]=V([g,i])$. From the hypothesis we have $V([g,i])\geq \alpha V^*([g,i])$ and $V([g,i])\geq A[g-1,j-1]\geq \alpha V[g-1,j-1]$. Hence, $A[i,j]=V([g,i])\geq \alpha \max\{V^*([g,i]), V[g-1,j-1]\}=\alpha V[i,j]$.
If $g\in [h,h^+]$ then $V([g,i])\leq V^*[i,j]$, by definition.
If $g\in[h^-,h-1]$ then $A[g,j-1]\leq V^*[i,j]$ so $V([g,i])\leq A[g,j-1]\leq V^*[i,j]$.
It is easy to see that it is not possible that $g=h^--1$. So in any case $\alpha V[i,j]\leq A[i,j]\leq V^*[i,j]$. 
There are a couple of more corner cases, for example what if $h^-$ or $h^+$ is not defined (or $h^-=h$, $h^+=h$), however these are special cases and can be handled with the same ideas.

In the end we have $\alpha V[m,k]\leq A[m,k]\leq V^*[m,k]\Leftrightarrow V[m,k]\leq \frac{1}{\alpha}V^*[m,k]$. The lemma follows. 
\end{proof}
The proof for the slower dynamic programming (quadratic on $m$) follows easily using the same ideas.

From Lemmas~\ref{lem:Obs}, \ref{proof1}, \ref{lem:AVG} we get the following results for $d=1$. For COUNT queries the partition with the minimum error can be found in $\mathcal{O}(m\log m)$ time. For SUM queries we can get a $2\sqrt{2}$-approximation of the optimum partition in $\mathcal{O}(km\log m)$ time. For AVG queries we can get a $2$-approximation of the optimum partition in $\mathcal{O}(km\log^2 m)$ time.

We note that the results for SUM and COUNT queries hold assuming that the ratio $\frac{N_i}{n_i}$ is the same for all possible valid buckets. If this is not the case and the ratios of two different buckets differ (at most) by a multiplicative parameter $\beta$ then the approximation factor of the partition we found needs to be multiplied by $\sqrt{\beta}$.
By constructing large enough partitions as shown in Section~\ref{sec:ratio} we guarantee that $\beta\leq 9$ with high probability, so all our approximation factors from the previous paragraph for COUNT and SUM queries should be multiplied by $\sqrt{9}=3$.

\mparagraph{Approximation for $d>1$}
We run the construction algorithm of Section~\ref{sec:kd-tree} considering a uniform set of $m$ samples.
Let $T'$ be the tree that our algorithm returns and let $T^*$ be the optimum k-d tree. For a node $u$ of a k-d tree let $V(u)$ be the variance of the query with the maximum variance in $u$. For a node $u$ let $V'(u)$ be the variance that our approximation algorithm reports as the maximum variance among the samples in $u$, satisfying $V'(u)\geq \alpha V(u)$.

\begin{lemma}
It holds that $V(T')\leq \frac{1}{\alpha}V(T^*)$. The algorithm for constructing $T'$ runs in $\mathcal{O}(m\log m+k\mathcal{H})$. time.
\end{lemma}
\begin{proof}
We note that if $u$ is an ancestor of a node $v$ of a k-d tree then $V(u)\geq V(v)$ for all types of queries COUNT, SUM, AVG. The proof of this argument is the same with the proof for 1D where the maximum variance of a partition $b_i$ is larger than the maximum variance of a partition $b_j$ where $b_j\subseteq b_i$. Hence, the overall maximum variance of $T'$ is non-increasing as we run more iterations constructing the tree.

If $T'$ is identical to $T^*$ then $T'$ is optimum. Next we consider the case where $T'$ is not identical to $T^*$.
In this case, there is always a leaf node $u$ of $T^*$, where $u$ belongs in $T'$ and $u$ is not a leaf node of $T'$. In other words, $u$ is a leaf node in $T^*$, however in one of the iterations of our algorithm we found that $u$ had the query with the largest (approximated) variance and we constructed its children.
Since $u$ is a leaf node of $T^*$ we have $V(T^*)\geq V(u)$.
At the moment that our algorithm decided to create the children of $u$ let $w$ be the leaf of $T'$ that contains the query with the real maximum variance. We have i) $V(u)\geq V'(u)\geq \alpha V(u)$, $V(w)\geq V'(w)\geq \alpha V(w)$ from the approximation algorithm for computing the maximum variance in a partition, ii) $V'(u)\geq V'(w)$ because the algorithm decided to split the node $u$ instead of node $w$, and iii) $V(u)\leq V(w)$ from the definition. It follows that $$V(T')\leq V(w)\leq \frac{1}{\alpha}V'(w)\leq \frac{1}{\alpha}V'(u)\leq \frac{1}{\alpha}V(u)\leq \frac{1}{\alpha}V(T^*).$$

A k-d tree over $m$ items can be constructed in $\mathcal{O}(m\log m)$ time. In each iteration we need to find an approximation of the maximum variance of $\mathcal{O}(1)$ new leaf nodes which can be done in $\mathcal{O}(k\mathcal{H})$ time. Furthermore, we need to store and update a Max-Heap of size $\mathcal{O}(k)$ so that we can find the next leaf node with the approximated maximum variance in constant time. This can be done in $\mathcal{O}(k\log k)$ time. Overall, our algorithm runs in $\mathcal{O}(m\log m + k(\mathcal{H}+\log k))=O(m\log m + k\mathcal{H})$ time, since $m>k$.
\end{proof}

From Lemmas~\ref{lem:Obs}, \ref{proof1} and Section~\ref{appndx:AVG} we get the following results for any constant $d>1$. For SUM and COUNT queries we get a $2\sqrt{k}$-approximation of the optimum k-d tree in $\mathcal{O}(m\log m)$ time. We notice that we do not use the running time from Section~\ref{appndx:SUM} since by constructing a balanced k-d tree over all sampled items we can find all the necessary sums in $\mathcal{O}(m\log m)$ time without constructing a range tree. For AVG queries in $2$-dimensions we can get a $2\sqrt{2k}$-approximation of the optimum k-d tree in $\mathcal{O}(km\log^3 m)$ time. For any dimension $d$ for AVG queries we can get a $\frac{\sqrt{2k}}{\delta^{1/2-1/(2d)}}$-approximation of the optimum k-d tree in $\mathcal{O}(km\log m)$ time. Using the range tree construction as we described in the end of Section~\ref{appndx:AVG} we can achieve the same approximation factor in $O(m\log^{d-1}m + \frac{k}{\delta^{1-1/d}}\log^{d} m)$ time for AVG queries.

\vspace{0.4em}

As we had in 1D, we note that the results for COUNT and SUM queries hold assuming that the ratio $\frac{N_i}{n_i}$ is the same for all possible valid buckets. If this is not the case then the approximation factor of the partition we found needs to be multiplied by $\sqrt{\beta}$.
By constructing large enough partitions as shown in Section~\ref{sec:ratio} we guarantee that $\beta\leq 9$ with high probability, so all our approximation factors for COUNT and SUM queries from the previous paragraph should be multiplied by $\sqrt{9}=3$.